\documentclass[a4paper,twocolumn,10pt,superscriptaddress,accepted=2019-05-09]{quantumarticle} 
\usepackage[colorlinks=true,linkcolor=blue,citecolor=blue,urlcolor=blue]{hyperref}
\pdfoutput=1

\usepackage[utf8]{inputenc}
\usepackage[backend=bibtex,sorting=none,firstinits=true,style=phys,maxbibnames=5,biblabel=brackets,chaptertitle=false,pageranges=false,doi=false]{biblatex}
\usepackage{dcolumn}   
\usepackage{bm}        
\usepackage{subfigure}
\usepackage{amsthm}
\usepackage[german, english]{babel}
\usepackage[export]{adjustbox}
\usepackage[]{units}
\usepackage{titlesec}
\usepackage{amsmath}
\usepackage{amsfonts}
\usepackage{amssymb}
\usepackage{braket}
\usepackage{lipsum}

\newtheorem{Theorem}{Theorem}

\newtheorem{Lemma}{Lemma}

\newtheorem{Prop}{Proposition} 

\begin{document}

\title{Subsystem symmetries, quantum cellular automata,  and computational phases of quantum matter}
\date{\today}
\author{David T. Stephen}
\affiliation{Max-Planck-Institut f{\"u}r Quantenoptik, Hans-Kopfermann-Stra{\ss}e 1, 85748 Garching, Germany}
\author{Hendrik Poulsen Nautrup}
\affiliation{ Institut f{\"u}r Theoretische Physik, Universit{\"a}t Innsbruck, Technikerstr. 21a, A-6020 Innsbruck, Austria}
\author{Juani Bermejo-Vega}
\affiliation{\mbox{Dahlem Center for Complex Quantum Systems, Freie Universit{\"a}t Berlin, 14195 Berlin, Germany}}
\author{Jens Eisert}
\affiliation{\mbox{Dahlem Center for Complex Quantum Systems, Freie Universit{\"a}t Berlin, 14195 Berlin, Germany}}
\author{Robert Raussendorf}
\affiliation{Department of Physics and Astronomy, University of British Columbia, Vancouver, British Columbia V6T 1Z1, Canada}
\affiliation{Stewart Blusson Quantum Matter Institute, University of British Columbia, Vancouver, British Columbia, V6T 1Z4, Canada}

\begin{abstract}
Quantum phases of matter are resources for notions of quantum computation. In this work, we establish a new link between concepts of quantum information theory and condensed matter physics by presenting a unified understanding of symmetry-protected topological (SPT) order protected by subsystem symmetries and its relation to measurement-based quantum computation (MBQC). The key unifying ingredient is the concept of quantum cellular automata (QCA) which we use to define subsystem symmetries acting on rigid lower-dimensional lines or fractals on a 2D lattice. Notably, both types of symmetries are treated equivalently in our framework. We show that states within a non-trivial SPT phase protected by these symmetries are indicated by the presence of the same QCA in a tensor network representation of the state, thereby characterizing the structure of entanglement that is uniformly present throughout these phases. By also formulating schemes of MBQC based on these QCA, we are able to prove that most of the phases we construct are computationally universal phases of matter, in which every state is a resource for universal MBQC. Interestingly, our approach allows us to construct computational phases which have practical advantages over previous examples, including a computational speedup. The significance of the approach stems from constructing novel computationally universal phases of matter and showcasing the power of tensor networks and quantum information theory in classifying subsystem SPT order.
\end{abstract}

\maketitle

The fields of study of quantum phases of matter and of quantum computation have been evolving alongside each other for over a decade, such that they are now deeply intertwined. This is on the one hand because many instances of non-trivial quantum order are key to storing or processing quantum information, an idea that has stimulated a plethora of theoretical and experimental research in both fields.
Perhaps the most familiar example is the idea of topological quantum computation which leverages the anyonic excitations of topologically ordered systems to perform error-resilient quantum computation~\cite{Kitaev2003,Freedman2003,Nayak2008}. A complementing approach to topological quantum computation uses Majorana fermions located at the edges of one-dimensional (1D) chains with symmetry-protected topological (SPT) order \cite{Kitaev2001,Alicea2011,Lutchyn2018}. 
On the other hand, quantum states exhibiting SPT order can be used as resources for instances of measurement-based quantum computation  \cite{Doherty2009,Miyake2010,Bartlett2010,Else2012,Else2012a,Miller2015,Wang2016,Stephen2017,Raussendorf2017,Nautrup2015,Miller2016,Miller2018,Wei2017,Wei2018,Raussendorf2018,Devakul2018a} -- an insight most relevant to the present work. Various further examples can be found \cite{Vijay2015,Dua2018,Bauer2018}; indeed, every time a new type of quantum order is discovered, it is not long before its uses in notions of quantum computation are being investigated. 

\begin{figure}[t]
\centering
\includegraphics[width=0.7\linewidth]{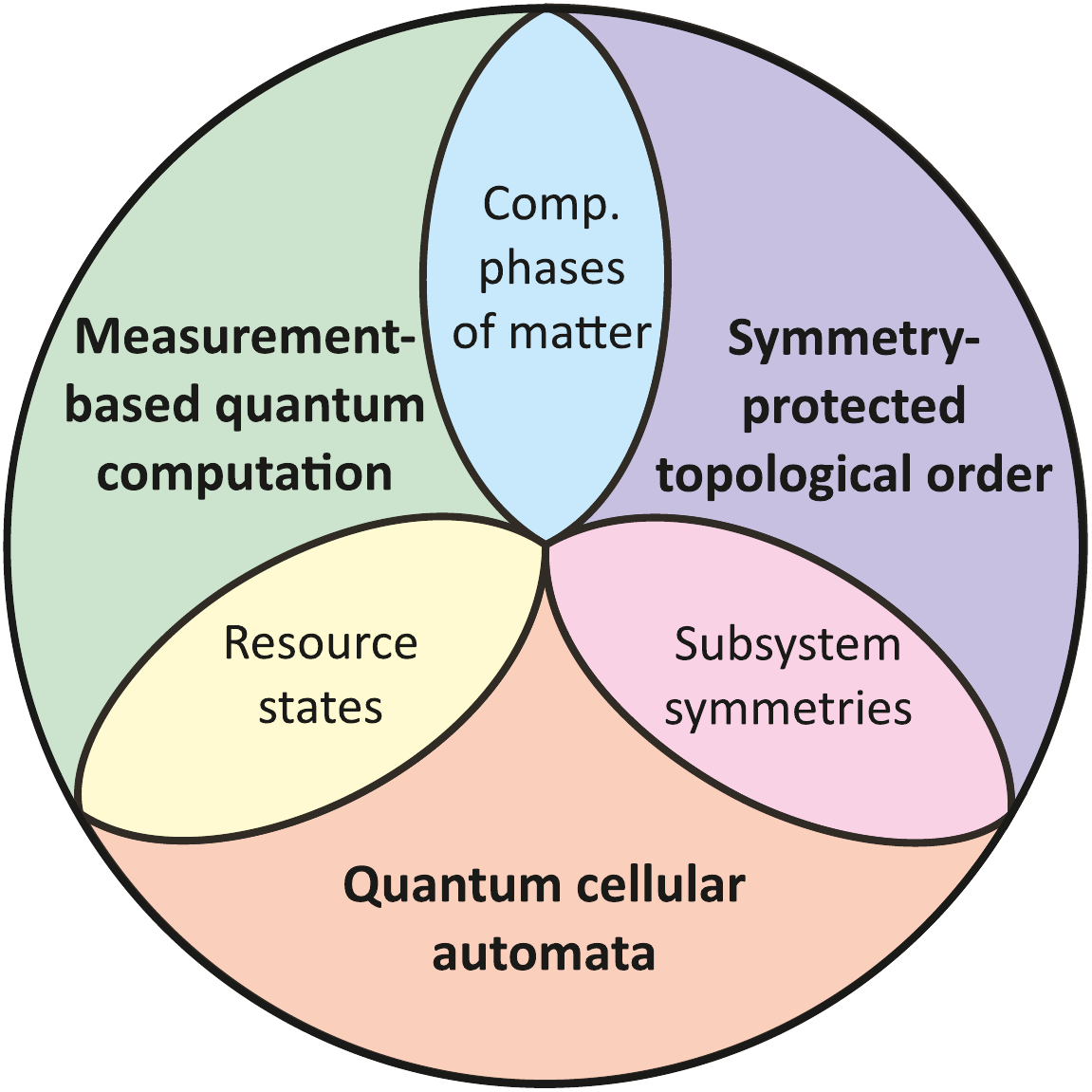}
\caption{Interrelation between symmetry-protected topological (SPT) order, measurement-based quantum computation (MBQC), and quantum cellular automata (QCA).  By framing both subsystem SPT order and MBQC in terms of QCA, we develop a new framework for characterizing subsystem SPT order and constructing computationally universal phases of matter.}
\label{fig:venn}
\end{figure}

The most intriguing aspect of this joint development, however, is that the 
intertwinement of the two fields goes both ways. Often, in-depth characterizations of quantum phases of matter are spurred by the desire to understand their computational properties \cite{Kitaev2003,0706.3401,Nayak2008,BombinDelgado2008,Yoshida2015,Roberts2017,Else2012}. 
In extreme cases of this, novel notions of quantum order have been defined in such a way that they are tailor-made for specific computational tasks \cite{Doherty2009,Else2012a,Miller2018,Raussendorf2018}.
 One such example can be found in Ref.~\cite{Raussendorf2018} in which the authors define a new two-dimensional (2D) SPT phase referred to as the ``cluster phase'', whose defining symmetries had also been suggested in earlier works~\cite{Doherty2009,Else2012a}. This phase is designed to be useful for measurement-based quantum computation (MBQC), a scheme of quantum computation that is performed by local measurements only on initially entangled states \cite{Raussendorf2001,Raussendorf2003}, avoiding any physical  implementation of coherent quantum gates. In the context of MBQC, a many-body quantum state is called a \textit{universal} resource state if the circuit model of quantum computation can be efficiently simulated by performing local measurements on the state. Hence, identifying and classifying universal resources is an important open problem in the study of MBQC. Recently, significant progress in this direction has been made thanks to a link between MBQC and SPT order: Many known MBQC resource states exhibit SPT order and, in some cases, their usefulness as resources persists throughout the entire SPT phase in which they lie. This link has nearly completely been established in 1D  \cite{Doherty2009,Miyake2010,Bartlett2010,Else2012,Miller2015,Wang2016,Stephen2017,Raussendorf2017}, but to a much lesser extent in 2D  \cite{Nautrup2015,Miller2016,Miller2018,Wei2017,Wei2018,Devakul2018a}. The aforementioned cluster phase has been the first example of a \textit{computationally universal phase of matter}, one in which every state is a universal resource for MBQC.

What distinguishes the cluster phase from other candidates for computationally universal phases? This answer lies in the structure of the symmetries that define it: While conventional SPT phases are defined by global on-site symmetries \cite{Pollman2010,Schuch2011,Chen2011,Chen2013,Senthil2015,PhysRevB.96.235138} or crystalline symmetries \cite{Song2017,Thorngren2018,Huang2017}, the symmetries in the cluster phase act rather on rigid 1D lines spanning the 2D lattice.
It is the first example of a \textit{subsystem SPT phase}, which is any SPT phase protected by symmetries that act on \textit{rigid} lower-dimensional structures, such as lines \cite{You2018,Devakul2018b}, planes \cite{You2018a}, or even fractals of non-integer dimension \cite{Williamson2016,Devakul2018,Kubica2018,Devakul2018a}, see Fig.~\ref{fig:fractal} for an example\footnote{Subsystem symmetries should be distinguished from \emph{higher-form} symmetries which act on \textit{deformable} lower-dimensional subsystems, are related to higher-form gauge theories~\cite{Gaiotto2015}, and have recently appeared in the context of quantum memories and quantum error correction~\cite{Yoshida2016,Kubica2018,Roberts2017,Roberts2018}.}. 
For such phases the total symmetry group, and therefore the dimension of the symmetry-protected edge states in which logical information is encoded and processed via MBQC, grows with system size. This is the key feature of subsystem SPT phases that enables their computational universality.  Independent from their use in MBQC, subsystem symmetries have recently garnered interest in condensed matter physics thanks to their relation to fracton topological order \cite{Haah2011,Yoshida2013,Vijay2015b,Vijay2016} (and, much earlier, topological order~\cite{Nussinov2009,Nussinov2009a}). It has been shown that, via generalized gauging procedures, certain states with subsystem symmetry can be transformed to states exhibiting fracton topological order~\cite{Vijay2016,Williamson2016,Kubica2018,You2018a,Shirley2018,Song2018}. Thus, it is high time to develop a framework for constructing and characterizing subsystem SPT phases.

\begin{figure}[t]
\centering
\subfigure[]{\label{fig:fractal}\includegraphics[width=0.8\linewidth]{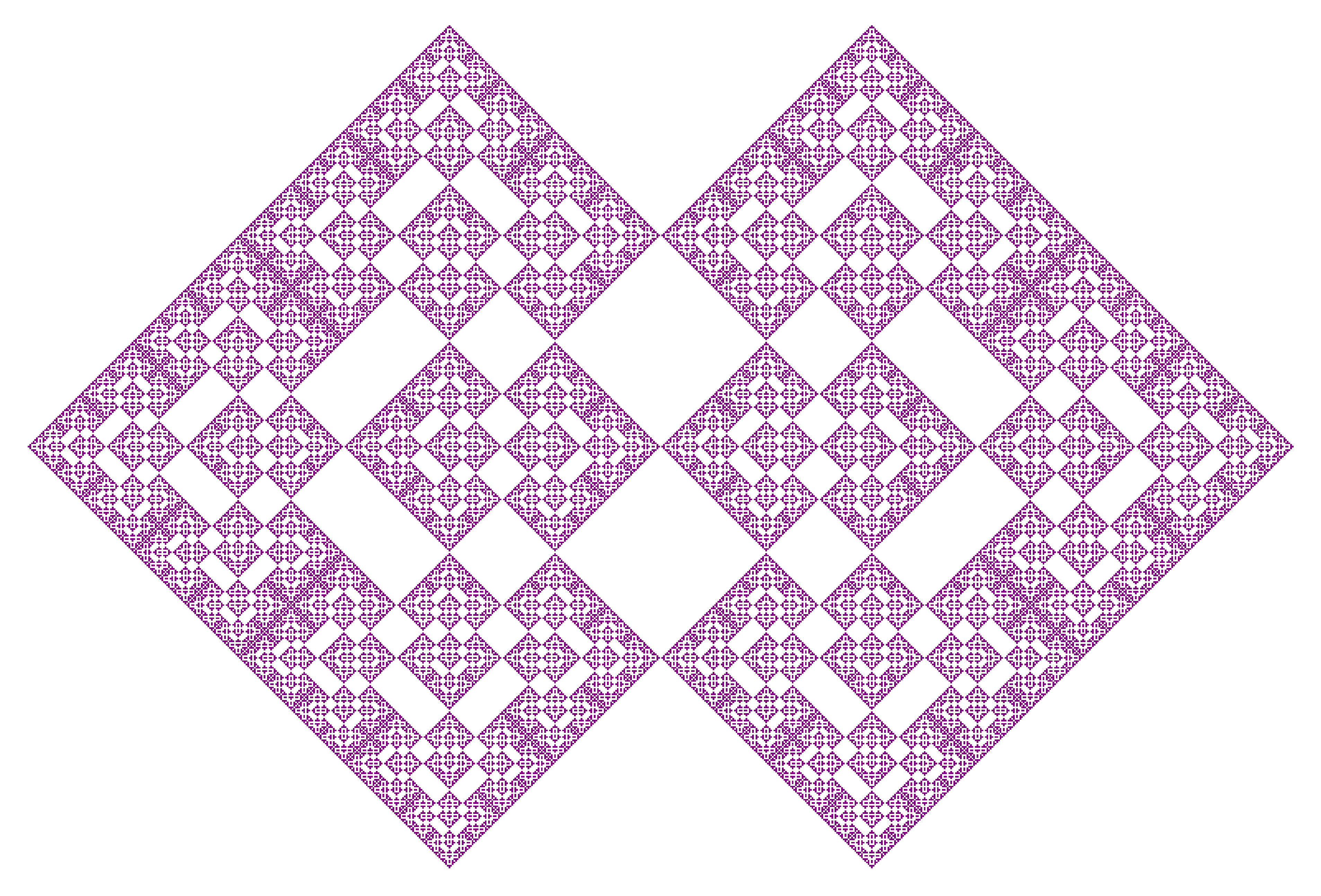}}
\subfigure[]{\label{fig:circuit}\includegraphics[width=0.8\linewidth]{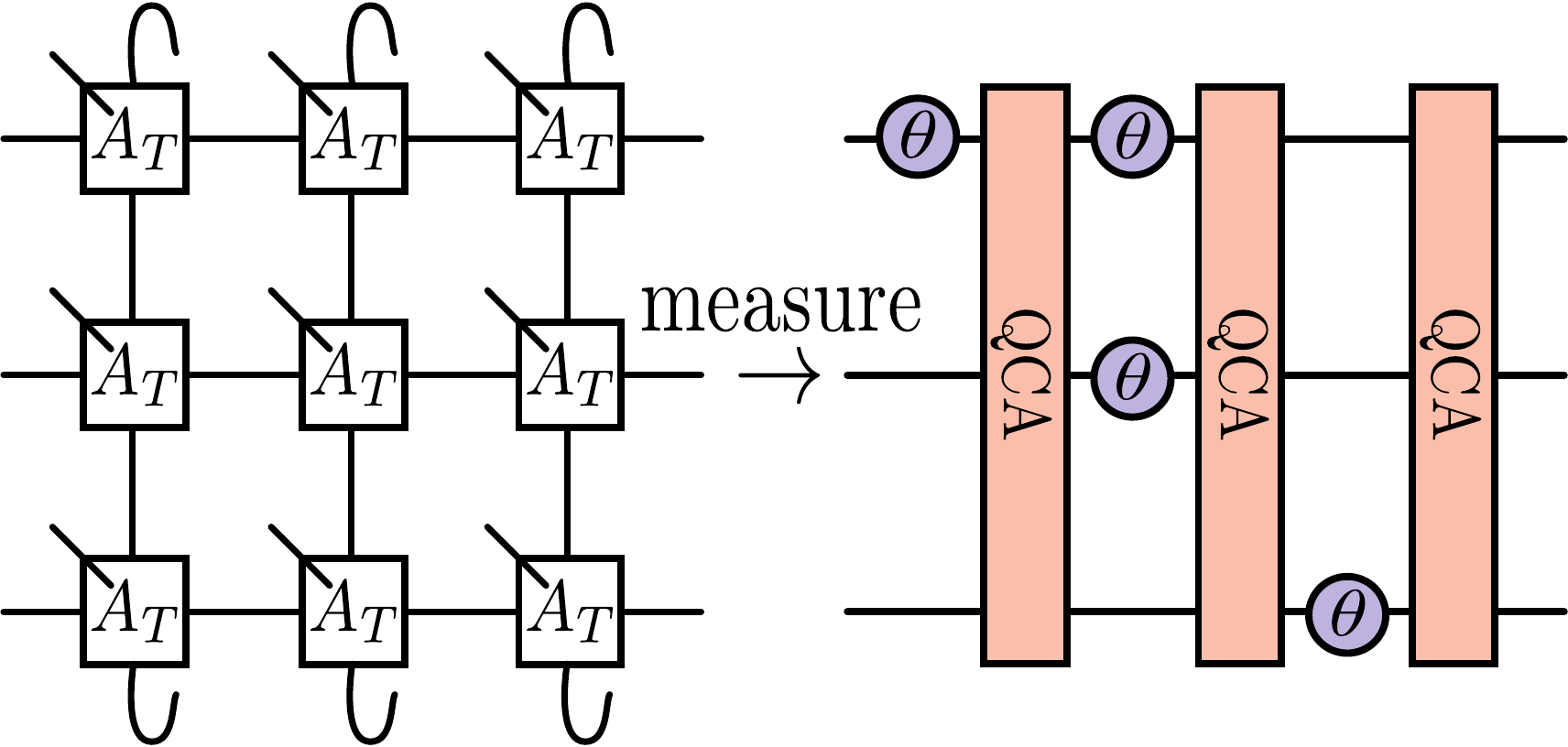}}
\caption{Two of the roles played by QCA in our results. a) Example of a fractal symmetry operator defined by QCA that emerges from our framework. The operator is a tensor product of local Pauli-$X$ operators arranged on a square lattice in the pattern shown. b) We define tensor network states from QCA, such that measuring a block of spins results in a quantum circuit living in the virtual space of the tensor network. The circuit, which forms the backbone of our MBQC scheme, consists of global applications of the QCA interspersed by single qubit rotations determined by the bases in which spins are measured. This structure appears within every state in the corresponding subsystem SPT phase, leading to uniform computational power across the phase.}
\label{fig:fractalcircuit}
\end{figure}

Historically, an important tool for developing the link between quantum phases of matter and quantum computation is the language of tensor networks \cite{Bridgeman2017}. In this language, the pure quantum state is decomposed into a network of local tensors associated to each site in the lattice. Remarkably, properties of a state that are inherently non-local, such as topological degeneracy \cite{Schuch2010} or string order \cite{Pollmann2012a}, can be understood locally in terms of symmetries of the local tensors. Studying these symmetries leads to an extremely powerful approach to detect and even classify topological order \cite{Schuch2010,Bultinck2017,Duivenvoorden2017}, SPT order \cite{Chen2011,Schuch2011,Pollman2010,Williamson2016a,Jiang2017,Molnar2018a}, and others \cite{Poilblanc2015,Williamson2017,Bultinck2018,Dreyer2018}. The same symmetries also aid in understanding a quantum state's usefulness for quantum computation. This is especially true when considering SPT order and MBQC, wherein the group of tensor symmetries that define 1D SPT order is essentially the only ingredient needed to construct a scheme of MBQC, and hence the possible MBQC schemes inherit the same algebraic structure that classifies 1D SPT order \cite{Stephen2017,Raussendorf2017}. 
In Ref.~\cite{Raussendorf2018} as well, it is shown that the cluster phase is characterized by a set of symmetries of the local tensor, and these symmetries were used to derive the computational scheme therein. Using similar methods, the same result has recently been shown for a pair of subsystem SPT phases protected by fractal symmetries \cite{Devakul2018a}. These results therefore point towards the possibility of a unified understanding of subsystem SPT order and its relation to MBQC via tensor networks.

In the present work, we provide such a unified framework by viewing both subsystem SPT order and MBQC in terms of underlying quantum cellular automata (QCA), thereby exposing a tight-knit relation between the three concepts, see Fig.~\ref{fig:venn}. The QCA are the new ingredient that becomes essential for the description 
of MBQC in SPT phases in dimensions higher than 1, see Fig.~\ref{fig:fractalcircuit} for illustrations. More precisely, we use QCA to define tensor network states that exhibit SPT order with respect to either rigid line-like or fractal subsystem symmetries, and we then characterize the corresponding subsystem SPT phases and their computational capabilities.

Our first result shows that non-trivial subsystem SPT order under the symmetries we consider is characterized by the presence of the QCA within the tensor network, which is persistent throughout the corresponding phase. This is akin to the behavior seen in the cluster phase, as the single-tensor symmetries that characterize the cluster phase can be used to derive the presence of a QCA in the tensor network \cite{Raussendorf2018}.
This means the patterns of entanglement found in these phases are characterized in part by QCA, demonstrating the possible use of tensor networks in obtaining a classification of subsystem SPT order. Interestingly, our framework treats line-like symmetries and fractal symmetries on the same footing, showing that different types of subsystem symmetries are more similar than one would think based on their structure. Indeed, we find that the 2D cluster state has subsystem SPT order under both types of symmetries.

We then turn to investigating the computational capability of the constructed phases. Using the above characterization, we show that every phase we construct is computationally universal in the same way as the cluster phase, except for those defined by non-entangling QCA. Hence, our framework gives a systematic way to identify computationally universal phases of matter which, up until now, have remained elusive outside a select few cases \cite{Raussendorf2018,Devakul2018a}. The computational schemes we develop are strictly tied to the QCA that define the phases (see Fig.~\ref{fig:circuit}), further strengthening the connection between quantum computation and SPT phases in 2D. 

Our perspective on MBQC based on QCA has the additional feature that, by choosing different QCA, the set of gates executable in a single step can be tailored to suit the problem at hand. In particular, our framework allowed us to uncover a particular class of subsystem SPT phases for which the corresponding computational schemes enjoy a quadratic reduction in the number of measurements per gate versus the number of logical qubits, as compared to previous schemes \cite{Raussendorf2018,Devakul2018a}. One such phase is built around a modified cluster state with additional qubits placed on the horizontal edges. We briefly discuss the implications that this result and our general framework may have on related tasks such as blind quantum computation~\cite{Fitzsimons2017,Mantri2017b}, quantum computation with global control~\cite{Raussendorf2005,Fitzsimons2006,Fitzsimons2007}, and in experimental demonstrations of  quantum computational advantage \cite{Bermejo-Vega2018,Hangleiter2018}.

This work is organized as follows. In Sec.\ \ref{sec:qca}, we begin with a review the basic properties of QCA. In Sec.\ \ref{sec:qcapeps}, we use QCA to define tensor network states and show that they have non-trivial SPT order under certain subsystem symmetries. Then, in Sec.\ \ref{sec:cyclesymms}, we investigate the properties of the corresponding SPT phases in a quasi-1D picture before moving to a genuine 2D picture in Sec.\ \ref{sec:sspt}. In Sec.\ \ref{sec:computation}, we classify these phases by their computational power in measurement-based quantum computing. Finally, in Sec.\ \ref{sec:discussion}, we discuss possible applications and extensions of our results.

\color{black}

\section{Quantum cellular automata} \label{sec:qca}

In this section, we present a review of \emph{quantum cellular automata (QCA)} for qubit systems, as described in Refs.~\cite{Schumacher2004,Schlingemann2008,Gutschow2010}, as they will be central to our
description of symmetry protected topological order with respect to
subsystem symmetries. A 1D QCA is a translationally-invariant locality-preserving unitary acting on a 1D chain of qubits~\cite{Schumacher2004}. That is, a QCA maps any locally supported operator to another locally supported operator, with the size of the support increased by an amount independent of the size of the original support. In Ref.~\cite{Cirac2017}, it was shown that QCA acting on 1D systems are equivalent to \emph{matrix product unitaries (MPU)}, in that every QCA can be represented as an MPU with finite bond dimension, and every MPU is a QCA. An MPU is a matrix product operator defined by a local tensor $\mathcal{T}$, which generates a unitary $T$ on a ring of arbitrary length $N$. Graphically, the MPU can be represented by the tensor network
\begin{equation}\label{eq:mpu}
T=\includegraphics[scale=1,valign=c]{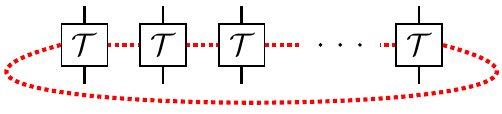}.
\end{equation}
\noindent Here, the dotted lines represent the virtual indices of the MPU. All QCA described in this paper act on systems with periodic boundary conditions. Note that we use the graphical language of tensor networks extensively throughout this work, see Ref.~\cite{Bridgeman2017} for a review of the relevant notation and concepts.

In what follows, we will focus on \emph{Clifford quantum cellular automata (CQCA)} \cite{Schlingemann2008,Gutschow2010}, which are QCA that map products of Pauli operators to products of Pauli operators. For reasons of simplicity of notation, we focus on the Pauli Clifford group
for qubits, even
though the formalism laid out here could be applied larger dimensions 
as well~\cite{Gottesman98Fault_Tolerant_QC_HigherDimensions,Bermejo-Vega2014}. 
On a finite chain of $N$ qubits, we can define the Pauli group $\mathcal{P}_N$ as the group generated by all local Pauli operators $X_i$, $Y_i$, and $Z_i$ acting on qubit $i$, where $i$ is defined modulo $N$ and can take negative values giving, for example, $X_{-1}=X_{N-1}$.  A CQCA is a QCA that is also an automorphism of $\mathcal{P}_N$, defined by a unitary transition function $T$ such that $P\mapsto T(P):= T^\dagger PT$ for any $P\in \mathcal{P}_N$.

CQCA have been studied extensively in Refs.~\cite{Schlingemann2008,Gutschow2010} which introduced a compact representation of CQCA. Firstly, a CQCA $T$ is completely specified by the images $T(X_0)$ and $T(Z_0)$. To see this, note that the QCA is translationally invariant, and that $X_i$ and $Z_i$ generate the whole Pauli group up to phases, which contains a basis of the space of all $2^N\times 2^N$ matrices.  Next, we represent elements of $\mathcal{P}_N$, up to phases, by $2N$-component binary strings $\bm{\xi}=(\bm{\xi}^X,\bm{\xi}^Z)$ such that 
\begin{equation} \label{eq:pauli}
V(\bm{\xi})=\bigotimes_{i=1}^N X_i^{\xi^X_i} Z_i^{\xi^Z_i}\in \mathcal{P}_N.
\end{equation} 
These strings form the group $\mathbb{Z}_2^N\times\mathbb{Z}_2^N$, for which $\bm{\xi}\mapsto V(\bm{\xi})$ forms a faithful irreducible projective representation. Note that we ignore all complex phases in front of Pauli operators throughout, since they do not affect the symmetries or phases of matter we define. We further condense the notation using the language of Laurent polynomials~\cite{Schlingemann2008}. We map $\bm{\xi}$ onto a vector of polynomials of a variable $u$ as 
\begin{equation}
\begin{pmatrix} \bm{\xi}^X  \\ \bm{\xi}^Z \end{pmatrix} \mapsto \begin{pmatrix} \sum_i (u^i)^{\xi^X_i}  \\ \sum_i (u^i)^{\xi^Z_i} \end{pmatrix}.
\end{equation}
The purpose of the variable $u$ is the keep track of the moving and spreading of local Pauli operators under $T$. For example, the operator $X_1Z_0X_{-1}$ is represented as 
\begin{equation}
\bm{\xi}=\begin{pmatrix} u+u^{-1}  \\ 1 \end{pmatrix}. 
\end{equation} 
We use the symbol $\bm{\xi}$ to represent both the binary and polynomial representations of an element of $\mathcal{P}_N$ interchangeably.

Finally, we can represent the CQCA $T$ as a $2\times 2$ matrix $t$ of polynomials by arranging 
\begin{equation}
T(X_0):= t\begin{pmatrix} 1 \\ 0\end{pmatrix},\,\,
T(Z_0):= t\begin{pmatrix} 0 \\ 1\end{pmatrix} 
\end{equation} 
into columns of a matrix. For a concrete example, consider the CQCA $T_g$ defined by the relations 
\begin{equation} \label{eq:tg}
T_g(X_i)=X_{i-1}Z_iX_{i+1}\,, T_g(Z_i)=X_i.
\end{equation} 
This CQCA is the QCA of the 2D cluster state~\cite{Raussendorf2003}, and it has appeared several times already in the context of quantum computation~\cite{Raussendorf2005,Bermejo-Vega2018,Mantri2017}. The subscript $g$ refers to the term \textit{glider} which we shall introduce shortly. In the polynomial representation, this CQCA becomes
\begin{equation}
t_g=\begin{pmatrix} u+u^{-1} & 1 \\ 1 & 0 \end{pmatrix}.
\end{equation}
Every CQCA $T$ can be represented as a 2 $\times$ 2 matrix $t$ whose entries are Laurent polynomials over $\mathbb{Z}_2$, up to phase factors~\cite{Schlingemann2008}. 
We further restrict to CQCA for which the images $T(X_i)$, $T(Z_i)$ are symmetric about site $i$, meaning there is no translation in the CQCA\footnote{In Refs.~\cite{Gross2012,Cirac2017}, QCA are assigned an index according to the amount of information flow to the left or right. Our restriction is to QCA with index 0}. As shown in Ref.~\cite{Schlingemann2008}, this corresponds to the matrix $t$ having unit determinant. We make this restriction only because translation will not be particularly interesting for our purposes. With this restriction, we have that all entries in $t$ are symmetric Laurent polynomials, meaning that $u^{-k}$ appears whenever $u^k$ does, for all $k$. Finally, we implement the periodic boundary conditions by taking all polynomials modulo the relation $u^N=1$.

CQCA can be split into three classes depending on their trace~\cite{Gutschow2010}.
\begin{itemize}
\item \textit{Periodic CQCA.} When $\mathrm{Tr}(t)=0$ or 1 the CQCA has \textit{periodic }behavior.
\item \textit{Glider CQCA.} When $\mathrm{Tr}(t)=u^{c}+u^{-c}$ for some positive integer $c$, the CQCA supports \textit{gliders}. These are operators on which the CQCA acts as translation by $\pm c$ sites. 
\item \textit{Fractal CQCA.} If neither of these conditions hold, the CQCA will display self-similar \textit{fractal} behavior. 
\end{itemize}
The CQCA in Eq.~(\ref{eq:tg}) is of glider type, with $\mathrm{Tr}(t_g)=u+u^{-1}$. Indeed, we can check that $T_g(X_iZ_{i-1})=X_{i+1}Z_i$, and $T_g(Z_{i+1}X_i)=Z_iX_{i-1}$.

The Laurent polynomial representation allows us to uncover an identity which will be useful at several points throughout this work. Namely, due to the Cayley-Hamilton theorem, we obtain~\cite{Gutschow2010}
\begin{equation} \label{eq:cayley}
t^2=\mathrm{Tr}(t) t + \mathbb{I},
\end{equation}
where we have used our assumption that $\det(t)=1$, and the fact that the polynomials are defined over the field $\mathbb{Z}_2$, so addition and subtraction are equivalent. This useful equation allows us to reduce any power of $t$ to a linear combination of $t$ and $\mathbb{I}$. 

\section{Defining PEPS from QCA} \label{sec:qcapeps}

Now we use the correspondence between QCA and MPU to define \emph{projected entangled pair states (PEPS)}~\cite{Verstraete2008}. Given a CQCA $T$, we first represent it as an MPU with local tensor $\mathcal{T}$, as described in Ref.~\cite{Cirac2017}. We then define the PEPS in terms of $\mathcal{T}$ by a local tensor $A_T$ whose components are given by,
\begin{equation}\label{eq:qcapeps}
\includegraphics[scale=1,valign=c,raise=0.13cm]{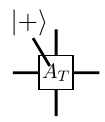}:=
\includegraphics[scale=1,valign=c]{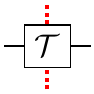}\quad,\quad
\includegraphics[scale=1,valign=c,raise=0.13cm]{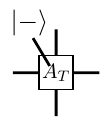}:=
\includegraphics[scale=1,valign=c]{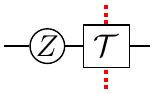}.
\end{equation}
Given ${\cal T}$, the PEPS tensors $A_T$ for local qubit dimension can be
uniquely defined in this way, the vectors $\{\ket{\pm}\}$ constituting a basis.
The resulting PEPS may not be rotationally invariant; the dimension of the virtual indices along the horizontal and vertical directions may not even match. Since we will be mainly treating these PEPS as quasi-1D systems, in a way that will be explained subsequently, we will not worry about this property of our construction here. We call the PEPS constructed in Eq.~(\ref{eq:qcapeps}) ``fixed-point'' PEPS because they will appear as special points within SPT phases, such as the AKLT state~\cite{Affleck1987} within the 1D Haldane phase. However, they are not  fixed-points of any renormalization transformation defined here.

It is important to stress how the term \textit{quasi-1D system} will be used here:
The mapping to a quasi-1D system occurs by putting our PEPS on a long, skinny torus of dimensions $N\times M$ ($M\gg N$), and blocking tensors $A_T$ into rings along the skinny direction of the torus. This gives rise to an effectively one-dimensional system,
\begin{equation} \label{eq:ringblock}
\includegraphics[scale=1,valign=c]{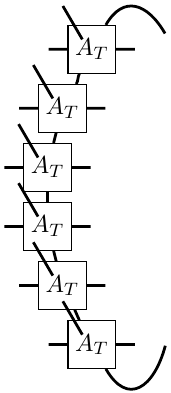}:=\includegraphics[scale=1,valign=c,raise=0.12cm]{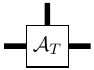} \quad .
\end{equation}
Therein, the contracted virtual legs correspond to the dotted red legs in Eq.~(\ref{eq:qcapeps}). The resulting ring tensor, denoted by $\mathcal{A}_T$, becomes the tensor of a 
\emph{matrix product state (MPS)} \cite{Perez-Garcia2006} representation of a quasi-1D system. For clarity, we use thicker lines when graphically representing MPS tensors like $\mathcal{A}_T$ which represent blocked tensors. We stress that, although we use a quasi-1D approach, the subsystem SPT phases that we construct can be distinct from stacks of 1D SPT chains, as will be discussed in Sec.~\ref{sec:sspt}.

If we contract the physical (upwards-directed) leg of the tensor $\mathcal{A}_T$ with the $N$-qubit state vector  ${|\mathbf{j}\rangle}={|j_1,j_2,\dots,j_N\rangle}$, the result is a $2^N\times 2^N$ matrix acting in the virtual space of the tensor network, denoted $\mathcal{A}_T^\mathbf{j}$. The state vector of the fixed-point PEPS is then written as
\begin{equation} \label{eq:mps}
{|\psi_T\rangle}=\sum_{\mathbf{j}_1,\dots\mathbf{j}_M} \mathrm{Tr}\left( \mathcal{A}_T^{\mathbf{j}_1}\mathcal{A}_T^{\mathbf{j}_2}\dots\mathcal{A}_T^{\mathbf{j}_M}\right) {|\mathbf{j}_1,\mathbf{j}_2,\dots,\mathbf{j}_M\rangle}.
\end{equation}
As is common in the description of quantum many-body systems with tensor networks~\cite{Schuch2011,Chen2011,Else2012,Else2012a,Miller2015,Raussendorf2017,Stephen2017}, throughout this work
we consider quantum states that can be captured as exact tensor network states, 
which is meaningful as they approximate general states arbitrarily well \cite{Schuch2008,Verstraete2006}. We also assume that our MPS can be constructed with \textit{injective} MPS tensors $\mathcal{A}_T$. This condition means that, for sufficiently large $l$, the set of products $\{\mathcal{A}_T^{\mathbf{j}_1}\dots\mathcal A_T^{\mathbf{j}_l}\}_{\mathbf{j}_1,\dots,\mathbf{j}_l}$ spans the space of all $2^N\times 2^N$ matrices~\cite{Perez-Garcia2006}. Physically, injectivity is equivalent to the state vector ${|\psi_T\rangle}$ having a finite correlation length in the long direction, which is true throughout the phases we consider.

In principle, we could have chosen to define our PEPS by replacing the Pauli $Z$ in the $|-\rangle$-component of $A_T$ with an $X$ or $Y$. It turns out that these cases are already included in the current definition. For example, replacing the $Z$ with an $X$ is equivalent to conjugating $T$ by Hadamard gates on every leg, represented as $t\mapsto hth^{-1}$ where
\begin{equation} \label{eq:hadamard}
h=\begin{pmatrix} 0 & 1 \\ 1 & 0 \end{pmatrix}.
\end{equation}
Similarly, exchanging $Z$ with $Y$ is equivalent to conjugation by phase gates, $t\mapsto sts^{-1}$, where
\begin{equation} \label{eq:phasegate}
s=\begin{pmatrix} 1 & 0 \\ 1 & 1 \end{pmatrix}.
\end{equation}
Were we to also include a Pauli operator in the $|+\rangle$-component of $A_T$, this would be equivalent to one of the three cases discussed above, up to rephasing of $T$, which is unimportant to us. Hence, we can use the definition of $A_T$ given by Eq.~(\ref{eq:qcapeps}) without loss of generality.

For the time being, we restrict our attention to \textit{simple CQCA}, which we define to be CQCA that have the form
\begin{equation} \label{eq:simpleqca}
t=\begin{pmatrix} \mathrm{Tr}(t) & 1 \\ 1 & 0 \end{pmatrix}.
\end{equation}
We make this restriction based on the fact that only simple CQCA define MPS tensors $\mathcal{A}_T$ that are injective for all $N$, as proven in Appendix \ref{app:simpleqca}. In Sec. \ref{sec:computation} we will slightly alter the way in which we construct PEPS from CQCA, and in this construction we will no longer need to impose any restriction on the choice of CQCA.

As follows straightforwardly from its definition (Eqs.~(\ref{eq:qcapeps},\ref{eq:ringblock})), the tensor $\mathcal{A}_T$ has the symmetries
\begin{align}\label{eq:ringsymms}
\includegraphics[scale=0.9,valign=c]{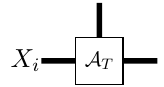}&=\includegraphics[scale=1,valign=c,raise=0.17cm]{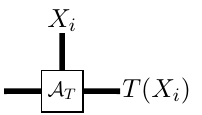}\ , \nonumber \\
\includegraphics[scale=1,valign=c]{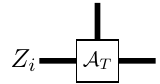}&=\includegraphics[scale=1,valign=c,raise=0.1cm]{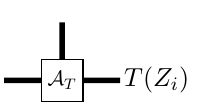}.
\end{align}
That is, a Pauli $Z$ acting in the virtual space passes through each ring freely and is transformed by the QCA $T$, while a virtual $X$ passes through with the help of a physical $X$ operator. Now, if $T$ represents a CQCA, $T(Z_i)$ and $T(X_i)$ are products of Pauli operators. Hence, they can be pushed through the next ring of tensors as well, possibly with the application of more physical $X$ operators. Since we impose periodic boundary conditions, the CQCA will have a finite period $L$ such that $t^L=\mathbb{I}$. In general, $L$ is a complicated function of the circumference $N$ of the torus, a point which we return to in Sec.~\ref{sec:period}. If we push any Pauli operator in the virtual space through $L$ rings of the PEPS, it will be mapped to itself, leading to the symmetry
\begin{align} \label{eq:qcacyclesymm}
&\hspace{0.75cm} \includegraphics[scale=1,valign=c,raise=0.35cm]{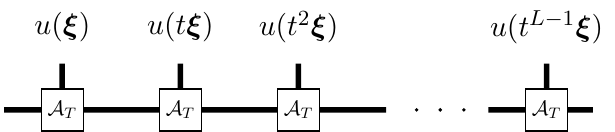} \nonumber \\
&\centerline{=} \nonumber \\
&\includegraphics[scale=1,valign=c]{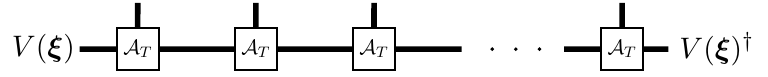},
\end{align}
for all $\bm{\xi}\in \mathbb{Z}_2^N\times\mathbb{Z}_2^N$, where 
\begin{equation} \label{eq:onsitesymm}
u(\bm{\xi}):=\bigotimes_{i=1}^N X_i^{\xi^X_i},
\end{equation} and $V(\bm{\xi})$ is defined in Eq.~(\ref{eq:pauli}). 
Hence, if we set $M=kL$, $k\in\mathbb{N}$, our state satisfies
\begin{equation}
U_T(\bm{\xi})^{\otimes k}{|\psi_T\rangle}={|\psi_T\rangle}
\end{equation}
 where the symmetry representation $\bm{\xi}\mapsto U_T(\bm{\xi})^{\otimes k}$ is defined as
 \begin{equation} \label{eq:cyclesymmdef}
 U_T(\bm{\xi}):= u(\bm{\xi})\otimes u(t\bm{\xi})\otimes\dots\otimes u(t^{L-1}\bm{\xi}). 
 \end{equation}
As shown in Appendix \ref{app:simpleqca}, $\bm{\xi}\mapsto U_T(\bm{\xi})$ is a faithful representation of $\mathbb{Z}_2^N\times\mathbb{Z}_2^N$ when $T$ is simple.

The unitary $U_T(\bm{\xi})$ is more general than the usual global on-site symmetry operator. Rather than acting the same way on each site in the lattice, the representation ``cycles'' with a period $L$. Such symmetries have been called \emph{$L$-cycle symmetries} in Ref.~\cite{Sauerwein2018}.
If we block our PEPS into large blocks of size $N\times L$, $U_T(\bm{\xi})^{\otimes k}$ becomes a standard global symmetry that acts in the same way on each block. In this way, we can look at the conventional 1D SPT order protected by $U_T(\bm{\xi})$. Since $\xi\mapsto V(\bm{\xi})$ forms a projective representation of $\mathbb{Z}_2^N\times \mathbb{Z}_2^N$, the fixed-point PEPS have non-trivial quasi-1D SPT order with respect to this symmetry~\cite{Pollman2010,Chen2011,Schuch2011}. As discussed in Sec.~\ref{sec:sspt}, this means they also have non-trivial 2D SPT order with respect to the same symmetries.

\vspace{5mm}
\noindent
\textit{Examples.} Let us study some examples of the fixed-point PEPS we have constructed. First, consider the CQCA $T_g$ from Eq.~(\ref{eq:tg}). In this case, the fixed-point PEPS defined by Eq.~(\ref{eq:qcapeps}) represents the 2D cluster state (see Appendix \ref{sec:stabilizers} for a proof of this and of the following examples). The corresponding $L$-cycle symmetry has the form of cone-like operators with $L=N$, as pictured in Fig.~\ref{fig:conesymms}. The SPT phase defined by these symmetries is exactly the cluster phase considered in Ref.~\cite{Raussendorf2018}. 

For an example using a CQCA with fractal behavior, consider the CQCA $T_f$ defined by the relations
\begin{equation} \label{eq:tf}
T_f(X_i)=X_{i-1}Y_iX_{i+1}\,, T_f(Z_i)=X_i.
\end{equation}
The corresponding fixed-point PEPS turns out to be the cluster state again, 
with the phase gate 
\begin{equation} \label{eq:sgate}
S=\begin{pmatrix} 1 & 0 \\ 0 & i \end{pmatrix}
\end{equation} applied to each site. The $L$-cycle symmetries also have a fractal structure, see Fig.~\ref{fig:fractal} for an example with $N=512$, $L=\frac{3}{2}N=768$. This shows that the cluster state also has SPT order under fractal symmetries that are tensor products of Pauli-$Y$ operators (since $SXS^\dagger=Y$).

For a periodic CQCA, we choose $T_p$ defined as
\begin{equation} \label{eq:tp}
T_p(X_i)=Z_i\,, T_p(Z_i)=X_i.
\end{equation}
That is, $t_p=h$ where $h$ is as defined in Eq.~\eqref{eq:hadamard}. This corresponds to a stack of decoupled 1D cluster states, and the $L$-cycle symmetries are simple horizontal lines. 

These examples are all states that can be defined by a local stabilizer group \cite{Gottesman1997}, which is a feature common for any fixed-point PEPS constructed by Eq.~(\ref{eq:qcapeps}), as is demonstrated in Appendix \ref{sec:stabilizers}. In fact, the fixed-point PEPS defined by simple CQCA are graph states \cite{Hein2004}, see Appendix \ref{app:simpleqca}. Hence, every state constructed by Eq.~(\ref{eq:qcapeps}) is the unique ground state of a gapped, exactly solvable Hamiltonian.

\begin{figure}
\centering
\includegraphics[scale=1.175]{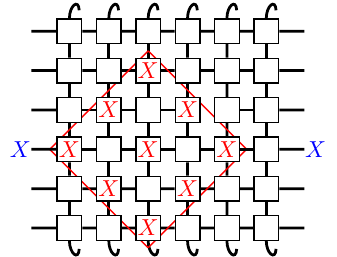} \hfill
\includegraphics[scale=1.175]{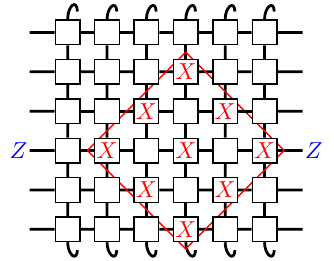}
\caption{Two generators of the symmetries of the cluster phase defined by the CQCA $T_g$ (Eq.~(\ref{eq:tg})) on a torus of circumference $N=6$. The other generators are vertical translates. The corresponding virtual representations for each generator are shown at the edges. Note that a Pauli $X$ operator drawn inside a PEPS tensor acts on the physical leg of the tensor, which is not explicitly shown.}
\label{fig:conesymms}
\end{figure}

\section{SPT order with $L$-cycle symmetries} \label{sec:cyclesymms}

Above, we have used CQCA to define fixed-point PEPS with quasi-1D SPT order protected by $L$-cycle symmetries. We would now like to investigate the corresponding SPT phases that surround these fixed-points. It is important to note that these are indeed gapped phases of matter, since the fixed-point PEPS have gapped parent Hamiltonians 
\cite{Schuch2011}
(with a uniform gap $\Delta> 0$ that is independent of the  system size, see Appendix \ref{sec:stabilizers}). 

To begin, we first need a better understanding of $L$-cycle symmetries in 1D systems.
In general, we consider a (quasi) 1D chain of $d$-level systems with length $M=kL$ which is invariant under $L$-cycle symmetries of the form $U(g)^{\otimes k}{|\psi\rangle}={|\psi\rangle}$ where
\begin{equation} \label{eq:mps2}
{|\psi\rangle}=\sum_{j_1,\dots,j_M=1}^d \mathrm{Tr}\left( \mathcal{A}^{j_1}\dots \mathcal{A}^{j_M}\right) {|j_1,\dots,j_M\rangle},
\end{equation}
and
\begin{equation} \label{eq:cyclesymm}
U(g):=u(g)\otimes u(\phi(g))\otimes\dots\otimes u(\phi^{L-1}(g)),
\end{equation}
with $g\mapsto u(g)$ being a $d$-dimensional unitary representation of a group $G\ni g$, and $\phi$ an automorphism of $G$ with $\phi^L=\mathbb{I}$. $g\mapsto u(g)$ alone need not be a faithful representation, but the whole cycle $U(g)$ should be faithful (otherwise $G$ should be redefined such that $U(g)$ becomes faithful). 
The symmetries found in the previous section fall under this definition.

Again, we can block $L$ consecutive sites into one larger site such that $U(g)^{\otimes k}$ acts in the same way on each block. 
If our state vector ${|\psi\rangle}$ is invariant under this symmetry, it is well-known that the blocked MPS tensors must satisfy the
relation~\cite{Perez-Garcia2008},
\begin{align} \label{eq:blockcyclesymm}
&\hspace{0.76cm} \includegraphics[scale=1,valign=c,raise=0.3cm]{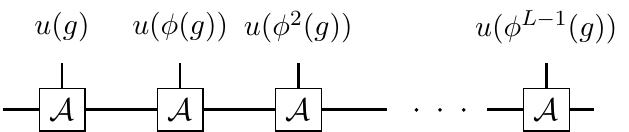} \nonumber \\
&\centerline{=} \nonumber \\
&\includegraphics[scale=1,valign=c]{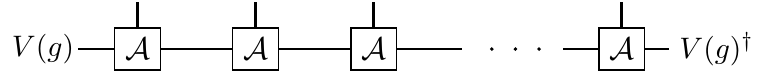} .
\end{align}
Therein, $g\mapsto V(g)$ is a projective representation of the group $G$, satisfying $V(g)V(h)=\omega (g,h)V(gh)$ for a cocycle 
$\omega $~\cite{Schuch2011}. The SPT order of ${|\psi\rangle}$ with respect to the symmetry $U(g)$ can be determined by $\omega$~\cite{Schuch2011}. For finite Abelian groups, which we will focus on here, there is a particularly important type of SPT phase called a maximally non-commutative phase~\cite{Else2012}. These phases satisfy the property 
\begin{align}
\{g|\omega(g,h)=\omega(h,g)\ \forall h\in G\}=\{e\}, 
\end{align}
and they have the largest edge-mode degeneracy of all phases for a given $G$~\cite{Marvian2017}. Another important property of these phases is that the decomposition of $V(g)$ into irreducible representations is made of up many copies of a single irrep $\widetilde{V}(g)$, such that we have the decomposition 
\begin{align}
V(g)=\mathbb{I}\otimes \widetilde{V}(g)
\end{align}
$\forall g\in G$~\cite{Berkovich1998}. 
Our first technical result is the following theorem.

\begin{Theorem}[Normal form of MPS in maximally non-commutative SPT phase] \label{theorem:1}
Any state vector ${|\psi\rangle}$ on a ring of length $kL$ that is in a maximally non-commutative SPT phase with respect to an $L$-cycle symmetry representation (Eq.~(\ref{eq:cyclesymm})) of a finite Abelian group $G$ admits an MPS representation of the form
\begin{equation} \label{eq:qcamps}
\mathcal{A}_{[l]}^j=B_{[l]}^j\otimes  (C^j\Phi),
\end{equation}
for suitable tensors $B_{[l]}^j$. Therein, $[l]$ is a site index such that $[l]=[l+L]$, $\Phi$ is uniquely defined by the relation $\widetilde{V}(g)=\Phi^\dagger \widetilde{V}(\phi(g))\Phi$, and $C^j=\widetilde{V}(g_j)$ for some $g_j\in G$. Throughout the phase, $C^j$ and $\Phi$ remain constant, while $B^j_{[l]}$ varies.
\end{Theorem}

This result is the essentially analogous to that given in Ref.~\cite{Else2012} when extended to $L$-cycle symmetries. $\Phi$ and $C^j$ are hence protected by the symmetry $g\mapsto U(g)$, in that they are present in a subspace of the virtual space of the MPS representation of all states in the phase. $C^j$ are completely defined by the on-site representation $g\mapsto u(g)$ of Eq.~(\ref{eq:cyclesymm}), while the transformation $\Phi$ contains the information about the structure of the $L$-cycle symmetry. Hence, the same patterns that define the $L$-cycle symmetry also appear in the entanglement structure.

\vspace{5mm}
\noindent
\textit{Proof.} We now sketch the proof of Theorem \ref{theorem:1}, with full details given in Appendix \ref{app:cycle}. The first step is to use the results of Ref.~\cite{Molnar2018} to show that the symmetries of the block tensor (Eq.~(\ref{eq:blockcyclesymm})) imply the following symmetries of the single tensor,
\begin{equation} \label{eq:singletenssycle}
\includegraphics[scale=1,valign=c,raise=0.25cm]{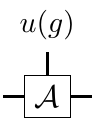}\ \propto\includegraphics[scale=1,valign=c]{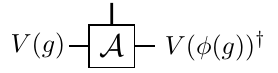}.
\end{equation}
Using the fact that the SPT phase is maximally non-commutative, we can then rewrite this equation as
\begin{equation} 
\includegraphics[scale=1,valign=c,raise=0.25cm]{singletenscycle_a.pdf}\ =\includegraphics[scale=1,valign=c]{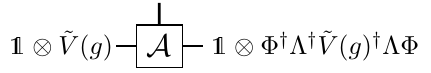},
\end{equation}
where $\Lambda$ is a matrix that encodes the proportionality constant that hides in Eq.~(\ref{eq:singletenssycle}), which it is not consistent throughout the phase. From this, we can continue as in Ref.~\cite{Else2012} to constrain the matrices $\mathcal{A}^j$ into the form
\begin{equation}
\mathcal{A}^j=B^j\otimes  (C^j\Lambda \Phi).
\end{equation}
Finally, we get rid of the non-universal $\Lambda$, at the cost of partially losing translation invariance of our MPS representation, and we are left with Eq.~(\ref{eq:qcamps}). $\Box$

\vspace{5mm}
We can now apply Theorem~\ref{theorem:1} to the constructions found in the previous section. Given a CQCA $T$, we get a fixed-point PEPS which is in a non-trivial quasi-1D SPT phase with respect to the $L$-cycle symmetry $\bm{\xi}\mapsto U_T(\bm{\xi})$ for $\bm{\xi}\in \mathbb{Z}_2^N\times \mathbb{Z}_2^N$. Since $V(\bm{\xi})$ are $N$-qubit Pauli operators, it follows that this SPT phase is maximally non-commutative. Hence, Theorem~\ref{theorem:1} applies, and if we consider an arbitrary state within this SPT phase, it admits a quasi-1D MPS representation of the form Eq.~(\ref{eq:qcamps}), which in this case reads
\begin{equation} \label{eq:qcamps4}
\mathcal{A}^\mathbf{j}_{[l]}=B^\mathbf{j}_{[l]}\otimes (C^\mathbf{j}T),
\end{equation}
where $\mathbf{j}=(j_1,j_2,\dots, j_N)$ denotes the state of the $N$ spins along a ring with $j=0$ (1) corresponding to the vectors 
${|+ \rangle}$ (${|-\rangle}$), and $C^\mathbf{j}=\bigotimes_{i=1}^N (Z_i)^{j_i}$. Notice that setting $B^{\mathbf{j}}_{[l]}=1$ leaves us with $\mathcal{A}^\mathbf{j}_{[l]}=C^\mathbf{j}T$, which is the fixed-point PEPS defined in Eq.~(\ref{eq:qcapeps}).

Thus, when CQCA are used to define subsystem symmetries, the presence of non-trivial SPT order under those symmetries is equivalent to the presence of the CQCA on the virtual level of the tensor network. This shows that the same structure appearing in the subsystem symmetry also appears in the entanglement structure found throughout the phase. This correspondence between QCA and SPT order protected by subsystem symmetries is the first major result of this work.

\section{Relation to subsystem SPT order} \label{sec:sspt}

In this section, we re-cast our phases, which are so far only defined in the quasi-1D picture, as genuine 2D phases of matter protected by subsystem symmetries. We then show that, for phases defined by glider CQCA, the full $L$-cycle symmetry group is not necessary to protect the phase, and a subgroup of rigid line-like symmetries is sufficient.

\subsection{From quasi-1D to 2D}

Let us be clear about what is meant by quasi-1D SPT order and how it differs from a genuine 2D SPT order. The difference is in the notion of locality. SPT order can be defined as equivalence classes of states under finite depth quantum circuits which respect the symmetry~\cite{Chen2013}; two states are in the same SPT phase if they can be transformed into each other by such a circuit, which corresponds to a quasi-adiabatic evolution along a path of gapped, local Hamiltonians \cite{RevModPhys.89.041004}. 
Here, `local' depends on the dimensionality of the system. For our quasi-1D scenario, the evolution must be local between different rings of the torus (the long direction of the torus), but it may be non-local along the rings (the skinny direction around the torus). When we promote to a true 2D scenario, we enforce locality in both directions. It is important to stress that the new 2D scenario again constitutes a proper gapped phase of matter, in that it embodies all states that can be reached from the gapped fixed point states via local quantum circuits that respect the symmetry. Using the techniques of Refs.\ \cite{Schuch2011,LowerBoundGap,Darwaman}, one can show a lower bound to the gap of the parent Hamiltonian in a vicinity of the fixed point model. Indeed, one can perturb the local tensors up to a constant error in the operator norm and remain in the same phase of matter.

From each fixed-point PEPS, we define the corresponding quasi-1D (2D) SPT phase as all states that can be reached from the fixed-point PEPS via symmetric finite depth quantum circuits which are local in the quasi-1D (2D) sense. Note that, since any circuit that is local in the 2D sense is also local in the quasi-1D sense, the 2D SPT phases are contained within the quasi-1D SPT phases. Therefore, the results on quasi-1D phases derived in the previous section, Eq.~(\ref{eq:qcamps4}) in particular, also hold throughout the corresponding 2D phase defined by the same symmetries. One implication of this is the presence of zero-energy edge modes of dimension $2^N$~\cite{Marvian2017} throughout the phase. This exponentially growing edge degeneracy was noted in Refs.~\cite{You2018,Devakul2018} as a signature of subsystem SPT phases. In Sec.\ \ref{sec:computation}, we encode logical information in these protected edge modes for the purposes of computation. 

We note an important distinction between the quasi-1D and 2D scenarios. As mentioned in Sec. \ref{sec:cyclesymms}, 1D SPT phases are classified by a symmetry group $G$, a representation $g\mapsto U(g)$, and a cocycle $\omega$. Often, the representation is not considered important to the classification of 1D SPT phases~\cite{Schuch2011, Chen2011,Verresen2017} because one can construct a path of Hamiltonians that smoothly interpolates between any two \textit{local} symmetry representations without closing the gap. So 1D SPT phases that only differ by the symmetry representation defining them are often said to be equivalent. This is precisely the case for the SPT phases considered here: When regarded as quasi-1D SPT phases, they differ only by the $L$-cycle symmetry representation $\bm{\xi}\mapsto U_T(\bm{\xi})$. Therefore, as quasi-1D phases, they may all be seen as equivalent. This argument breaks down in the 2D scenario, since the interpolation can act non-locally on blocks of size $N\times L$, violating the 2D notion of locality. So SPT phases defined by different CQCA can be distinct in the genuine 2D scenario. Indeed, as we saw above, defining the symmetry by the CQCA $T_p$ (Eq.~(\ref{eq:tp})) leads to an SPT phase built around a stack of decoupled 1D cluster states, while the phase defined by the CQCA $T_g$ (Eq.~(\ref{eq:tg})) is built around the 2D cluster state. These two phases are physically distinct, as discussed in the next subsection.

\subsection{Line-like symmetries protect glider CQCA}

In the case of glider CQCA, promoting to a 2D notion of locality allows us to equivalently define our phases in terms of rigid line-like symmetries acting on 1D subsystems. In general, the symmetries defined by glider CQCA are shaped like cones. This follows from 
Eq.~(\ref{eq:cayley}) which, for glider CQCA, takes the form $t^2=(u^c+u^{-c})t+\mathbb{I}$. Thus, if we apply $T$ to any single-qubit Pauli $P_i$, we find 
\begin{equation}
T(T(P_i))=T^2(P_i)=T(P_{i-c})P_iT(P_{i+c}).
\end{equation}
We can continue in this fashion to see that powers of $T$ can be expressed as cones which expand until they wrap all the way around the torus, after which they contract back to a point (see Fig.~\ref{fig:glidercones} for the case $P_i=Z_i$).

\begin{figure}[t]
\centering
\subfigure[]{\label{fig:glidercones}
\begin{adjustbox}{width={0.85\linewidth},totalheight={\textheight},keepaspectratio}
\begin{tabular}{|c|c|c|c|c|c|c|}
\hline
  &      &      &      &      &      &   \\ \hline
  &      &      & $T(Z)$ & $Z$    &      &   \\ \hline
  &      & $T(Z)$ & $Z$    & $T(Z)$ & $Z$    &   \\ \hline
~~$Z$~~  & $T(Z)$ & $Z$    & $T(Z)$ & $Z$    & $T(Z)$ & ~~$Z$~~  \\ \hline
  &      & $T(Z)$ & $Z$    & $T(Z)$ & $Z$    &   \\ \hline
  &      &      & $T(Z)$ & $Z$    &      &   \\ \hline
\end{tabular}
\end{adjustbox}}
\subfigure[]{\label{fig:gliderlines}
\begin{adjustbox}{width={0.85\linewidth},totalheight={\textheight},keepaspectratio}
\begin{tabular}{|c|c|c|c|c|c|c|}
\hline
     &      & $T(Z)$ & $Z$    &      &      &      \\ \hline
     & $T(Z)$ & $Z$    &      &      &      &      \\ \hline
$T(Z)$ & $Z$    &      &      &      &      & $T(Z)$ \\ \hline
$Z$    &      &      &      &      & $T(Z)$ & $Z$    \\ \hline
     &      &      &      & $T(Z)$ & $Z$    &      \\ \hline
     &      &      & $T(Z)$ & $Z$    &      &      \\ \hline
\end{tabular}
\end{adjustbox}}
\caption{Representation of the general appearance of cones (a) and lines (b) in glider CQCA, for $c=1$. Each column represents a product of Pauli operators on $N=6$ qubits with periodic boundaries, and advancing to the right is equivalent to one application of $T$. After $6$ steps, any operator returns to itself.}
\end{figure}

This picture also gives us an easy way to understand the existence of gliders, see Fig.~\ref{fig:gliderlines}. The gliders lead to line-like symmetries of the state, as pictured in Fig.~\ref{fig:linesymms}. The line symmetries can be generated as products of the cone symmetries, but the converse is not true on a finite torus. Indeed, the line symmetries form a $\mathbb{Z}_2^{2(N-1)}$ subgroup of the full $\mathbb{Z}_2^{2N}$ symmetry group of cones in the cluster phase.

\begin{figure}[t]
\centering
\includegraphics[scale=1.075,valign=c]{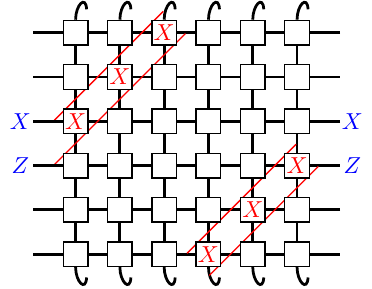}\hfill
\includegraphics[scale=1.075,valign=c]{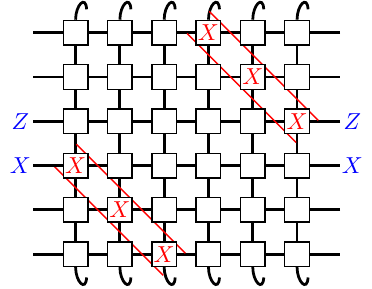}
\caption{The gliders of a CQCA define line-like symmetries of the PEPS, as pictured here for the CQCA $T_g$ of Eq.~(\ref{eq:tg}) which defines the cluster phase. The corresponding representation in the virtual space is shown on the edges.}
\label{fig:linesymms}
\end{figure}

The claim is that, when a 2D notion of locality is enforced, the line symmetries defined by gliders are sufficient to protect the SPT order. Fig.~\ref{fig:gliderlines} shows that the operator $T(Z_i)Z_{i-1}$ defines a glider and hence a line symmetry. Likewise, $Z_i T(Z_{i-1})$ defines a glider moving in the opposite direction. Now, for any even integer $k$, we can write 
\begin{align}
&Z_i\otimes Z_{i-k}= \nonumber \\
&(Z_iT(Z_{i-1}))(T(Z_{i-1})Z_{i-2})\dots (T(Z_{i-k+1})Z_{i-k}). 
\end{align}
That is, we can write a product of non-neighbouring $Z$'s as a product of gliders. Since $Z_i$ in the virtual space corresponds to a cone symmetry in the physical space, we see that we can create pairs of cone symmetries with products of line symmetries. The same argument can be repeated for $X_i$ in place of $Z_i$.

If we separate these cones sufficiently far from each other, then each local gate in the finite depth quantum circuit will see only one of the cones. So a local gate that is symmetric under the line symmetries is also symmetric under the cone symmetries. Since we require in the definition of SPT order that each \textit{local} gate of the finite depth quantum circuit commutes with the symmetry, this implies that a circuit which is symmetric under the line symmetries is also symmetric under the cone symmetries. Therefore, the 2D phase defined by line symmetries in contained within the quasi-1D phase defined by the cone symmetries.

Definitions of 2D SPT order with symmetries acting on 1D lines have also been suggested in Refs.~\cite{Else2012a,Raussendorf2018,You2018,Devakul2018b}. In Refs.~\cite{You2018} and \cite{Devakul2018b}, the authors make a distinction between \textit{strong} and \textit{weak} subsystem SPT order under line-like symmetries. Therein, a phase is called weak if it can be trivialized by adding 1D SPT chains along the direction of the symmetry operators, while strong subsystem SPT phases are genuine 2D phases of matter that cannot be viewed as stacks of 1D SPT chains. 
For simple, periodic CQCA, the fixed-point states we construct are themselves stacks of 1D chains, so periodic CQCA correspond to weak subsystem SPT phases. On the other hand, phases defined by simple, glider CQCA are strong. This has been shown in Ref.~\cite{Devakul2018b} for the 2D cluster phase. For phases defined by other simple glider CQCA, one can use the stabilizer representation of the fixed-point PEPS (Appendix \ref{sec:stabilizers}) to see that they correspond to stacks of 2D cluster states. Since the subsystem symmetry group acts independently on each cluster state, the resulting system is still in a strong subsystem SPT phase. Finally, for fractal CQCA, there are no line-like symmetries~\cite{Gutschow2010}, and we must treat the fractal symmetries themselves as fundamental. This leads to a definition of 2D phases via fractal symmetry operators, as in Refs.~\cite{Devakul2018,Kubica2018}. Whether there exists a notion of strong versus weak phases under fractal symmetry remains unknown~\cite{Devakul2018b}.

\section{Computational power of the QCA phases} \label{sec:computation}

Now that we have defined 2D SPT phases via CQCA, and we have understood some basic properties of these phases, we move on to characterizing their computational power in measurement-based quantum computing. 

Before we can state our second main result, we need to define another class of CQCA, \textit{entangling} CQCA.
We call a CQCA $T$ entangling if one or more of the entries of the matrix $t$ is not 0 or 1.
In other words, entangling CQCA are CQCA that spread information. All glider and fractal CQCA are entangling, but not all periodic CQCA are. In particular, non-entangling CQCA are those for which $t$ can be expressed as a product of $h, h^{-1}$ and $s, s^{-1}$ as defined in Eq.~(\ref{eq:hadamard}) and Eq.~(\ref{eq:phasegate}), respectively. For this class of CQCA, we will prove our second main result:

\begin{Theorem}[Computational phases of matter] \label{theorem:2}
For every entangling CQCA, there exists a 2D SPT phase in which every state is a resource for universal MBQC, except for a possible subset of zero measure. Furthermore, the universal circuit model is simulated with polynomial overhead. 
\end{Theorem}

We begin with an outline of the main ideas that are needed to use an SPT phase as a resource for MBQC \cite{Stephen2017, Raussendorf2017}. There are three key features of the constructions which are relevant to the current work:
\begin{enumerate}
\item Logical qubits are encoded in the virtual space of the tensor network~\cite{Gross2007,0706.3401}. Specifically, $N$ qubits are encoded in the $2^N$-dimensional subspace of Eq.~(\ref{eq:qcamps4}) which is uniform throughout the corresponding SPT phase. One time-step of the computation is enacted by individually measuring every qubit in a \textit{block} of size $N\times L$ consisting of $L$ consecutive rings around the torus, where $L$ is the period of the CQCA (to be determined in Sec. \ref{sec:period}). Information can be initialized into this space and subsequently read-out by appropriate measurement patterns on blocks.

\item Logical gates are performed by measuring a single qubit in a block in a perturbed basis 
\begin{equation} \label{eq:tiltedbasis}
\{ {|+\rangle} + id\alpha {|-\rangle},{|-\rangle}-id\alpha {|+\rangle} \},
\end{equation}
with $|d\alpha|\ll 1$, and measuring the rest of the qubits in a block in the basis $\{ {|+\rangle},{|-\rangle} \}$. According to Eq.~(\ref{eq:qcamps4}), and the computational scheme of Ref.~\cite{Raussendorf2017}, if the qubit measured in the perturbed basis is located at site $(i,l)$ in the block, the corresponding logical gate is, up to second order in $d\alpha$,
\begin{equation} \label{eq:gates}
R_{(i,l)}(d\alpha)=\exp(2id\alpha \nu_{(i,l)} T^{L-l+1}(Z_i)),
\end{equation}
where $\{\nu_{(i,l)}\}$ is a set of constants that characterize the part of the state that is not uniform within the SPT phase. That is, they are defined in terms of the $B^\mathbf{j}_{[l]}$ from Eq.~(\ref{eq:qcamps4}). These constants can be easily measured before computation, again using local measurements only. Therefore, they may be accounted for by adjusting $d\alpha$ accordingly. If one of $\nu_{(i,l)}$ is equal to $0$, which only occurs for a subset of states of zero measure, the computation fails. The gate $R_{(i,l)}(d\alpha)$ represents an infinitesimal rotation generated by $T^{L-l+1}(Z_i)$. By composing these gates in the appropriate order, we can achieve any rotation generated by elements of the Lie algebra $\mathcal{O}_T$ which is generated by the set $\{T^l(Z_i)|i=1,\dots, N, l=1,\dots, L\}$ with linear combinations and the matrix commutator. Thus our full set of gates is given by the Lie group $\mathcal{L}_T=\exp(i\mathcal{O}_T)$. $\mathcal{L}_T$ is the same for every state in the SPT phase defined by the CQCA $T$ (to be investigated in Sec.\ \ref{sec:gates}).

\item Every non-trivial gate must be followed up by measuring a large number of blocks of qubits in the $\{ {|+\rangle},{|-\rangle}\}$ basis. The number of blocks measured is on the order of the correlation length of the system (in the long direction), which is finite by the assumption of injectivity. This serves to decouple the two virtual subspaces in Eq.~(\ref{eq:qcamps4}), which become slightly entangled after each logical gate.

\end{enumerate}
\noindent

\subsection{Period of the CQCA} \label{sec:period}

We first determine the period $L$ of our CQCA, which determines the length of a single step of computation. We recall the matrix representation $t$ of the CQCA. The task is to determine the smallest $L$ such that $t^L=\mathbb{I}$, which we do by invoking Eq.~(\ref{eq:cayley}). Since the trace of the CQCA appears here, the analysis now splits into three parts, depending on whether the CQCA is of periodic, glider, or fractal type.

\vspace{5mm}
\noindent
\textit{Periodic CQCA.} For periodic CQCA, we have $\mathrm{Tr}(t)=a$, where $a=0,1$. Using Eq.~(\ref{eq:cayley}), it is easy to check that the period is then $L=a+2$ \cite{Gutschow2010}. Note that this period is independent of the circumference $N$. In particular, the CQCA is periodic even on an infinite chain, while the glider and fractal CQCA have a finite period only when periodic boundary conditions are enforced. This is the reason why only these CQCA are called ``periodic''.

\vspace{5mm}
\noindent
\textit{Glider CQCA.} For CQCA that support gliders, the period can be determined from the cone structure in Fig.~\ref{fig:glidercones}. For $c=1$ and $N$ even, this cycle takes $N$ steps (for odd $N$, it takes $2N$ steps, but we only consider even $N$). For $c>1$, the above is still true, but there may be a smaller number $L$ also satisfying $t^L=\mathbb{I}$. We ignore this possibility for simplicity. Thus, the period of a glider CQCA on a ring of circumference $2N$ can always be taken to be $L=2N$.

\vspace{5mm}
\noindent
\textit{Fractal CQCA.} The case is more complicated for CQCA with fractal behaviour. Indeed, due to their fractal nature, the period $L$ of these CQCA can be a wildly fluctuating function of $N$. In fact, $L(N)$ can appear to have exponentially growing behavior (see Appendix \ref{app:fractal} for an example). 
This would pose a significant problem to computation. Since $N$ is essentially the number of qubits, and $L$ controls the duration of a single step of computation, an exponential relationship implies that computation time scales exponentially with the number of qubits. Thus, we could not call the resulting computational scheme universal, even if we have a full set of gates. Luckily, it turns out that this problem can be avoided. Indeed, although $L(N)$ may have an exponential envelope, it turns out that there is a subsequence of system sizes for which the relationship in linear. Specifically, when $N=2^k$, we have either $L=2^k=N$ or $L=2^k+2^{k-1}=\frac{3}{2}N$, as proven in Appendix \ref{app:fractal}.

\subsection{Determining gate set} \label{sec:gates}

We must now determine the Lie group of gates $\mathcal{L}_T$. We will show that we can construct a universal set of gates as long as $T$ is simple and entangling, which implies that $T$ is of glider or fractal type. We construct arbitrary single-qubit gates and a non-trivial two-qubit entangling gate, which together form a full gate set~\cite{Barenco1995,DeutschBarencoEkert1995,Bremner2002}. The single-qubit gates follow from the fact that $T$ is a simple CQCA, which implies that $T(Z_i)=X_i$. Then, if we set $l=1$ in Eq.~(\ref{eq:gates}), we get all $Z$-rotations on a single qubit, while setting $l=L$ gives all $X$-rotations. Together, these give a full set of single qubit gates. 

To construct an entangling gate, we generalize the technique from Ref.~\cite{Raussendorf2018}, in which a two-qubit gate was constructed from a three-qubit gate by initializing one qubit into a particular eigenstate on which the three-qubit gate acts trivially. Consider $T^2(Z_i)$ for any $i$. If $T$ is entangling, this operator must act non-trivially outside of site $i$. Then $T^2(Z_i)$ will be a product of Pauli operators supported on the interval $[i-n,i+n]$, for some minimal $n\neq 0$, which is symmetric about site $i$. The trick now is to make every $2n$-th qubit a logical qubit, and initialize the qubits in between each logical qubit into the +1 eigenstate of the middle $2n-1$ operators in $T^2(Z_i)$. This initialization can always be done since we have all single qubit gates at our disposal, as described in Ref.~\cite{Raussendorf2018}. Then, $T^2(Z_i)$ will act as a two-qubit gate on logical qubits at positions $i-n$ and $i+n$ while leaving the qubits in between unchanged. Hence, by setting $l=L-1$ in Eq.~(\ref{eq:gates}), we get a non-trivial two-qubit entangling gate on all neighbouring pairs of logical qubits which, together with the single-qubit gates, form a universal set.

For an explicit example, consider the CQCA from Eq.~(\ref{eq:tg}) with $T_g^2(Z_i)=X_{i-1}Z_iX_{i+1}$. We initialize every even numbered qubit in the ${|0\rangle}$ eigenstate and use only odd-numbered qubits as logical qubits. Then the gate $\exp(2id\alpha T_g^2(Z_{2i}))=\exp(2id\alpha X_{2i-1}Z_{2i}X_{2i+1})$ reduces to $\exp(2id\alpha X_{2i-1}X_{2i+1})$, which is an entangling gate on the logical qubits. 

\subsection{Proving computational universality} \label{sec:proof}
We have now all ingredients needed for the proof of Theorem \ref{theorem:2} for simple, entangling CQCA. Next we show how to modify the above scheme such that we can drop the condition that our CQCA is simple and instead apply it to any entangling CQCA. One advantage of using simple QCA is that they allow us to easily construct a full set of single-qubit gates. For arbitrary CQCA that are not simple, this is not as straightforward. To address this issue, we modify the way in which we construct fixed-point PEPS from CQCA. 

The new fixed-point PEPS are defined in terms of a two-qubit unit cell, labeling the two qubits within a unit cell by $a$ and $b$. The local tensor $A'_T$ has components given by,
\begin{align} \label{eq:2qubitqcapeps}
&\includegraphics[scale=1,valign=c,raise=0.24cm]{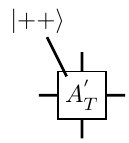}:=
\includegraphics[scale=1,valign=c]{qcapeps_I.pdf}\qquad , \quad
\includegraphics[scale=1,valign=c,raise=0.24cm]{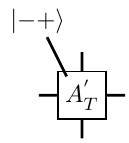}:=
\includegraphics[scale=1,valign=c]{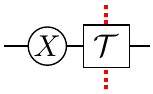},\nonumber \\
&\includegraphics[scale=1,valign=c,raise=0.24cm]{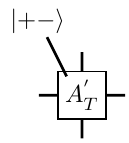}:=
\includegraphics[scale=1,valign=c]{qcapeps_Z.pdf}\,\, , \quad
\includegraphics[scale=1,valign=c,raise=0.24cm]{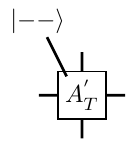}:=
\includegraphics[scale=1,valign=c]{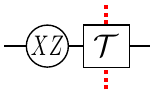}, 
\end{align}
where $\mathcal{T}$ is again the MPU representation of the CQCA. The choice of CQCA in this definition is completely free, as opposed to the previous sections where we required $T$ to be simple CQCA.

\begin{figure} 
\centering
\includegraphics[scale=1.078,valign=c]{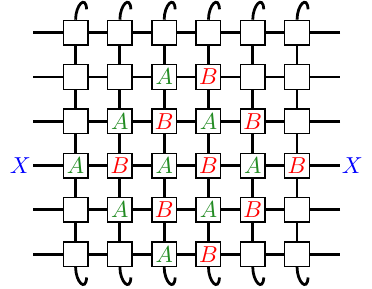}\hfill
\includegraphics[scale=1.078,valign=c]{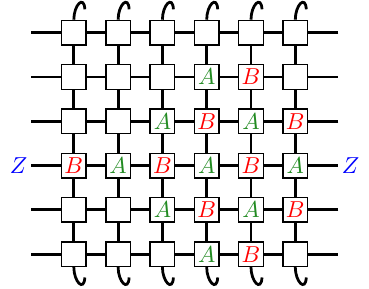} \\ \vspace{0.8cm}
\includegraphics[scale=1.078,valign=c]{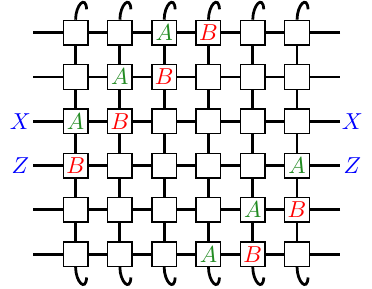}\hfill
\includegraphics[scale=1.078,valign=c]{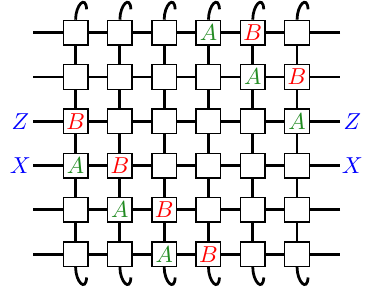}
\caption{The modified line and cone symmetries arising from the cluster CQCA $T_g$ of Eq.~(\ref{eq:tg}) when using the two qubit per-site construction. $A$ ($B$) denotes the Pauli-X operator acting on qubit $a$ ($b$) on a given site, while operators on the edges show the corresponding virtual representation. Note that the analysis in Sec.~\ref{sec:sspt} still holds in this case: the line symmetries, although widened, are still sufficient to protect the phase.}
\label{fig:2qubitsymms}
\end{figure}

The tensors $A'_T$ can be used to define the ring tensor $\mathcal{A}'_T$ as in Eq.~(\ref{eq:ringblock}). Proceeding from here in the same way as the qubit case, we can construct an $L$-cycle symmetry representation 
\begin{equation} \label{eq:twoqubitcycle}
\bm{\xi}\mapsto U'_T(\bm{\xi})=u'(\bm{\xi})\otimes u'(t\bm{\xi})\otimes\dots\otimes u'(t^{L-1}\bm{\xi})  
\end{equation}
of $\mathbb{Z}_2^N\times\mathbb{Z}_2^N$,
where the on-site representation is defined by 
\begin{equation} \label{eq:twoqubitonsite}
u'(\bm{\xi})=\bigotimes_{i=1}^N {X^a_i}^{\xi^X_i} {X^b_i}^{\xi^Z_i} ,
\end{equation}
where $X^a_i$, $X^b_i$ denotes the Pauli-X matrix acting on the qubit $a$ and $b$ at site $i$, respectively (see Fig.~\ref{fig:2qubitsymms}). It can be easily verified that $\bm{\xi}\mapsto U'_T(\bm{\xi})$ is a faithful representation of $\mathbb{Z}_2^N\times\mathbb{Z}_2^N$ as before.
The virtual representation $\bm{\xi}\mapsto V(\bm{\xi})$ is again the Pauli representation, so the resulting state has non-trivial SPT order under this $L$-cycle symmetry. We can thus prove that Eq.~(\ref{eq:qcamps4}) holds throughout the corresponding SPT phase, where $C^\mathbf{j}$ can now be any $N$-qubit Pauli, not just a product of $Z$'s. 

We can briefly compare this two-qubit construction to the original one-qubit construction of Eq.~(\ref{eq:qcapeps}). As shown in Appendix \ref{app:simpleqca}, every PEPS defined by Eq.~(\ref{eq:qcapeps}) with a simple CQCA $T$ is equivalent to the PEPS defined by Eq.~(\ref{eq:2qubitqcapeps}) with CQCA $T^2$. So this new construction includes all fixed-point PEPS, and hence all SPT phases, considered in the previous sections. Since we no longer restrict to simple CQCA, we also get new fixed-point states and phases. In particular, all fixed-point states in the previous sections were graph states, while this is not generally the case for states defined by Eq.~(\ref{eq:2qubitqcapeps}). They are nevertheless still stabilizer states (see Appendix~\ref{sec:stabilizers}), and so they are again unique ground states of gapped, local Hamiltonians.

The important distinction between the two constructions appears in their use for MBQC. Because of the different $C^\mathbf{j}$ appearing in Eq.~(\ref{eq:qcamps4}), the set of logical gates that we can execute in one step is enlarged. The computational scheme is almost identical to that described above for the single-qubit unit cell, except that we now have a choice to measure either qubit $a$ or $b$ at a given site in a perturbed basis for a logical gate. The logical gates we can execute now have the form,
\begin{equation} \label{eq:gates2}
R'_{(i,l)}(d\alpha)=\exp(2id\alpha \nu_{(i,l)} T^{L-l+1}(P_i)),
\end{equation}
where $P_i$ is either $Z_i$ or $X_i$, depending on whether qubit $a$ or $b$ is measured in the perturbed basis, respectively. Therefore, setting $l=1$ in Eq.~(\ref{eq:gates2}) already gives a full set of single-qubit rotations. We can then construct an entangling gate in a way exactly analogous to the previous case. Again, this construction can fail only if $T$ is not entangling.

With this, we have constructed a full set of gates for all CQCA $T$ which are entangling. So we can finally complete the proof of Theorem \ref{theorem:2},

\vspace{5mm}
\noindent
\textit{Proof of Theorem 2}. Based on the above discussions, we see that a) $N$ can always be scaled in such a way that the period $L$ is linear in $N$ or constant, and b) we can construct a universal set of gates on $q=\lfloor N/{2n}\rfloor $ qubits for a fixed $n$, \textit{i.e.} $SU(2^q)\subset\mathcal{L}_T$. Finally, the computational scheme of Refs.~\cite{Stephen2017,Raussendorf2017,Raussendorf2018} which we employ here has polynomial overhead for logical gates. Thus, overall, we can simulate the circuit model of quantum computation on $q$ logical qubits in $poly(q)$ time, as stated in the theorem. $\Box$

\vspace{5mm}

Within each phase, the computational power remains uniform and hence they truly constitute computational phases of matter. However, it is important to note that the measurement protocol needed for gates depends on certain details about the specific point within the phase, namely the constants $\nu_{(i,l)}$, as in similar schemes from previous works \cite{Stephen2017,Raussendorf2017,Raussendorf2018,Devakul2018a}. This means that our schemes are not necessarily robust against unknown small deformations within the phase, as would be the case in fault-tolerant schemes for quantum computing. However, the constants $\nu_{(i,l)}$ can be readily obtained via local measurements, in a small pre-computation before the actual quantum computation \cite{Raussendorf2017}, giving rise to uniform computational power across the phase. These constants are the same for each block of computation, as we have assumed translational invariance throughout.

\subsection{Periodic, entangling CQCA as universal resources for MBQC} \label{sec:periodicqca}

Theorem~\ref{theorem:2} does not apply to every periodic CQCA, and there are examples of periodic CQCA which do not lead to universal phases: consider again the CQCA $T_p$ in Eq.~(\ref{eq:tp}), which is periodic and describes decoupled 1D chains. Such a system cannot be a universal resource for measurement-based quantum computation with our methods, since we cannot create entangling gates with local measurements. This shows that the presence of zero-energy edge modes of dimension $2^N$ is not sufficient for universal MBQC, and that some additional structure is needed to allow entangling gates. However, periodic CQCA can also be entangling and hence computationally universal as stated in Theorem~\ref{theorem:2}. Surprisingly, in such cases being periodic is not a bug but rather a feature that we now investigate.

Periodic, entangling CQCA have a computational advantage over glider or fractal CQCA due to a quadratic reduction in the number of measurements per gate. This improvement stems from the constant period $L$ of periodic CQCA in contrast to a period that scales linearly in $N$ for glider or fractal CQCA (see Sec.~\ref{sec:period}).
As an example, consider the following periodic, entangling CQCA,
\begin{equation}\label{eq:entanglperiod}
t_{e}=ht_g=\begin{pmatrix} 1 & 0 \\ u+u^{-1} & 1 \end{pmatrix}.
\end{equation}
The corresponding fixed-point state defined via Eq.~(\ref{eq:2qubitqcapeps}) is a dressed cluster state which features additional qubits along horizontal lines (see Fig.~\ref{fig:dressedcluster}). The CQCA has a period of $2$ since $t_{e}^2=\mathbb{I}$. Nevertheless, it is entangling since $T(X_i)=Z_{i-1}X_{i}Z_{i+1}$ spreads information. According to our computational scheme $T(X_i)$ can be used as an entangling gate on qubits $i-1$ and $i+1$ while single qubit gates come for free in the two-qubit construction. The advantage of the periodic CQCA comes from the fact that our scheme requires all the qubits within a block to be measured in order to perform a single gate. Since the block size is $N\times L$ with $L=2$ for $t_{e}$, we gain a computational advantage over glider and fractal CQCA which require at least $N\times N$ measurements per block. 

\begin{figure}[t]
\centering
\includegraphics[scale=0.5,valign=c]{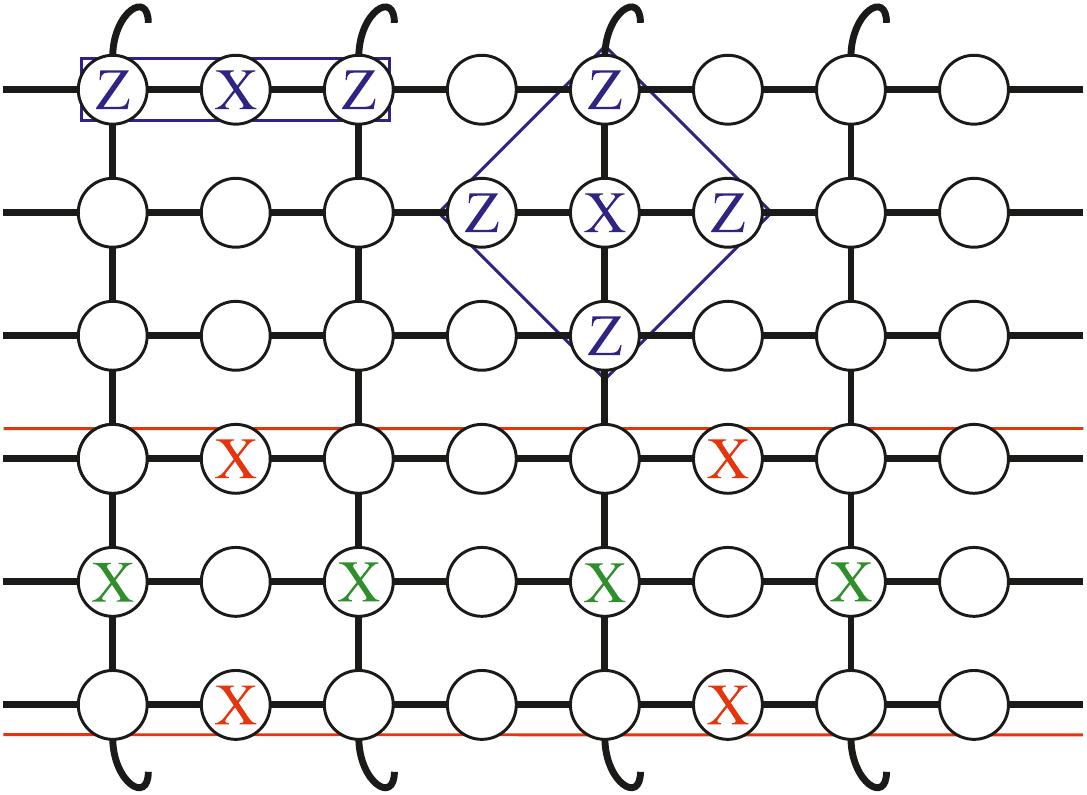}
\caption{The dressed cluster state as defined by the CQCA $t_{e}$ in Eq.~(\ref{eq:entanglperiod}). Each circle represents a qubit. Representative graph state stabilizers are depicted in blue. One symmetry generator is depicted by red and green $X$ operators acting on the $A$ and $B$ sublattice, respectively. Since $t_{e}$ is periodic, symmetries act along horizontal lines. However, because $t_{e}$ is also entangling, the width of some horizontal lines is larger than one.}
\label{fig:dressedcluster}
\end{figure}

\color{black}

\section{Discussion} \label{sec:discussion}

In this work, we have used quantum cellular automata to define subsystem symmetries and, correspondingly, SPT phases of matter. With this, we developed a new framework for identifying computationally universal phases of matter, and also for characterizing subsystem SPT order. We have determined which of the resulting phases of matter are computationally universal from the perspective of measurement-based quantum computation. The relation we uncovered between QCA and SPT order protected by subsystem symmetries in tensor networks should aid in the understanding and classification of these relatively novel phases of matter.

 Our general framework of building MBQC schemes based on CQCA remains relatively unexplored. In particular, we can choose different CQCA in order to tailor the elementary logical gate set towards the problem at hand, a fact that we did not take advantage of in our general proof of universality. For example, when two-qubit nearest neighbour gates are sufficient, periodic CQCA can be employed for a computational speedup as shown in the previous section. On the other hand, altering the elementary gate set in MBQC to include higher order entangling gates is possible in our framework, and can lead to depth-savings in gate synthesis
\cite{Gachechiladze2018}. Also, as pointed out in Ref.~\cite{Devakul2018a}, the self-similar nature of fractal CQCA can lead to the possibility of entangling far-separated qubits in a single step. It is therefore worthwhile to explore the different computational properties bestowed by different CQCA.

The flexibility of our framework could also have implications for other related protocols of quantum computation that are partially based around QCA. One example is that of secure delegated quantum computation: some protocols for universal blind quantum computation~\cite{Fitzsimons2017,Mantri2017b} employ schemes for MBQC that only make use of measurements in a single plane. This requirement is naturally fulfilled by our scheme. The same requirement has also recently appeared in proposals to demonstrate superior quantum computational power in near-term devices \cite{Bermejo-Vega2018,Hangleiter2018}.
Also related is quantum computation restricted to translationally invariant operations ~\cite{Raussendorf2005,Fitzsimons2006}, which is particularly useful for quantum computing architectures that are restricted to global control~\cite{Fitzsimons2007}. There are examples of all of the above protocols that are based on the cluster state QCA $T_g$. 
It would be interesting to investigate whether our general framework can be adapted to these settings, and in particular whether the speedup identified in Sec.~\ref{sec:periodicqca} carries over to advantages in any of the above protocols.

In the direction of classifying subsystem SPT order, one future path would make further use of the tensor network language. In Ref.~\cite{Raussendorf2018}, it was shown that the subsystem SPT order of the cluster phase is associated with a set of symmetries of a single tensor.  Likewise, the symmetries of a ring of tensors considered here can be used to determine the symmetries of the single tensor. Examining the symmetries of a single tensor has proven useful for classifying phases of matter and their physical properties in the past~\cite{Chen2011,Schuch2011,Schuch2010,Williamson2016a,Molnar2018a}, so it is likely this will be the case for subsystem symmetries as well. 

Finally, while we only considered QCA acting as Clifford circuits on chains of qubits, it is possible to generalize beyond this scenario. First, we could extend our analysis to systems of arbitrary local dimension $d$. In this case, we could realize systems with subsystem symmetry groups of the form $(G\times G)^N$ for an arbitrary finite Abelian group $G$ with $|G|=d$, and our analysis would be in terms of generalized Pauli and Clifford operators~\cite{Bermejo-Vega2014}. Our general formalism in Sec.~\ref{sec:cyclesymms} is already equipped to handle this extension, which would lead to new computationally universal phases, and may be useful for the classification of subsystem SPT phases. Similarly, we can extend our results to higher dimensions using the same mapping to a quasi-1D system. In particular, moving to 3D opens up the study of fracton topological order \cite{Haah2011,Yoshida2013,Vijay2016} and fault-tolerant MBQC \cite{Raussendorf2006,Raussendorf2007}. Finally, with some modifications, our framework should also be able to handle non-Clifford QCA, although in this case the resulting subsystem symmetries would likely not be simple products of local operators.

This work is supported by the NSERC (DTS, RR), Cifar (RR), the Stewart Blusson Quantum Matter Institute (RR), the European Union through the ERC grants TAQ (JBV, JE) and WASCOSYS (DTS), the DFG CRC 183  (JE),  and the Austrian Science Fund FWF within the DK-ALM: W1259-N27 (HPN). DTS thanks N.\ Schuch for discussions, and H.\ Dreyer for discussions and help in preparation of the manuscript. RR thanks the American Institute of Mathematics for their hospitality during the workshop ``Arithmetic golden gates'' in 2017.

\appendix

\section{Stabilizers for fixed-point PEPS} \label{sec:stabilizers}

In Sec.\ \ref{sec:qcapeps}, we build SPT phases around special fixed-point PEPS as defined by Eq.~(\ref{eq:qcapeps}). In order to gain a better intuition and to construct gapped parent Hamiltonians, we show in this section that the fixed-point PEPS are also stabilizer states with local stabilizer groups.

To start, notice that the fixed-point PEPS defined by Eq.~(\ref{eq:qcapeps}) have more symmetries than those shown in Eq.~(\ref{eq:ringsymms}). They also have the following symmetries,
\begin{equation}\label{eq:fixedptsymms}
\includegraphics[scale=1,valign=c]{ringsymm_c.pdf}=
\includegraphics[scale=1,valign=c,raise=0.16cm]{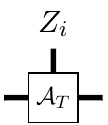}=
\includegraphics[scale=1,valign=c,raise=.08cm]{ringsymm_d.pdf}.
\end{equation}
Note that these symmetries exist only at the fixed-point, and they do not persist throughout the corresponding SPT phase. In addition to these symmetries, we also need the identity from Eq.~(\ref{eq:cayley}). In general, we have $\mathrm{Tr}(t)=\sum_{k=1}^m \alpha_k (u^{k}+u^{-k})+\beta$ for $\alpha_k,\beta\in\{0,1\}$. Stated in terms of operators, Eq.~(\ref{eq:cayley}) tells us, 
\begin{equation}\label{eq:t2}
T^2(Z_i)=Z_i\bigotimes_{k=1}^m \left[T(Z_{i-k})T(Z_{i+k})\right]^{\alpha_k}T(Z_i)^\beta.
\end{equation}
Using Eqs.~(\ref{eq:ringsymms}),~(\ref{eq:fixedptsymms}) and (\ref{eq:t2}), we can now determine the form of the stabilizers for our states,
\begin{align}
\includegraphics[scale=1,valign=c]{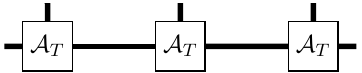}&=\includegraphics[scale=1,valign=c,raise=0.15cm]{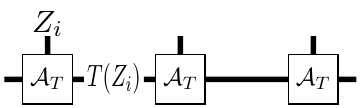}\nonumber \\
&=\includegraphics[scale=1,valign=c,raise=0.3cm]{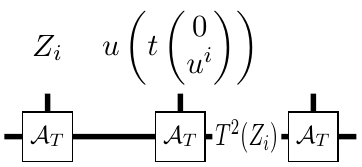} \nonumber \\
&=\includegraphics[scale=1,valign=c,raise=0.3cm]{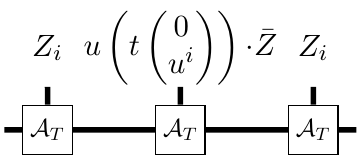}. 
\end{align}
In the first equality, we have used Eq.~(\ref{eq:fixedptsymms}). In the second, we have employed Eq.~(\ref{eq:ringsymms}) (note that $Z_i$ is represented as the polynomial vector $(0,u^i)$). In the third equality, we have used Eq.~(\ref{eq:t2}), and Eq.~(\ref{eq:fixedptsymms}) to pull the $Z_i$ and $\bar{Z}:= \bigotimes_{k=1}^m \left[Z_{i-k}Z_{i+k}\right]^{\alpha_k}Z_i^\beta$ onto the physical legs. Since $u(\xi)$ is a localized product of $X$ operators and $\bar{Z}$ is a localized product of $Z$ operators, we have derived a local stabilizer operator. We can derive an independent stabilizer $K_i$ for any site $i$ in the lattice. One can easily check that, due to the reflection symmetry of the stabilizers, all stabilizers commute as needed. So the fixed-point PEPS is the unique ground state of the Hamiltonian,
\begin{equation} \label{eq:stabham}
H=-\sum_i K_i,
\end{equation}
which is a local Hamiltonian with a uniform spectral gap. 

We can now verify the claimed examples in Section \ref{sec:qcapeps}. If we construct the stabilizers for the fixed-point PEPS defined by the CQCA $T_g$ in Eq.~(\ref{eq:tg}), then we find that they have the following form,
\begin{align}
&Z \nonumber \\
K_i = Z\  &X_i\  Z, \nonumber \\
&Z
\end{align}
which is indeed the familiar form of the cluster stabilizer \cite{Raussendorf2003}. If we use the CQCA $T_f$ in Eq.~(\ref{eq:tf}), we find
\begin{align}
&Z \nonumber \\
K_i = Z\  &Y_i\  Z, \nonumber \\
&Z
\end{align}
which corresponds to the cluster state with the operator $S$ (Eq.~(\ref{eq:sgate}))
applied to every site. Finally, with the CQCA $T_p$ of Eq.~(\ref{eq:tp}), we get stabilizers with the simple form,
\begin{equation}
K_i = ZX_i Z,
\end{equation}
which does indeed correspond to a stack of uncoupled 1D cluster states.

The fixed-point PEPS defined with respect to a two-qubit unit cell in Eq.~(\ref{eq:2qubitqcapeps}) are also stabilizer states with a local stabilizer group, as can be straightforwardly shown in the same way as above. Hence, these fixed-point PEPS are also ground states of uniformly gapped Hamiltonians.

\section{Simple CQCA} \label{app:simpleqca}

In this section we show that, when using the single-qubit unit cell construction of fixed-point PEPS (Eq.~(\ref{eq:qcapeps})), it is sufficient to consider simple CQCA, which we define to be any CQCA of the following form,
\begin{equation} \label{eq:simpleqca2}
t=\begin{pmatrix} \mathrm{Tr}(t) & 1 \\ 1 & 0 \end{pmatrix}.
\end{equation}
We then describe some properties of fixed-point PEPS defined by simple CQCA.

\subsection{Injectivity of fixed-point states}

Consider the $N$-qubit MPS tensor $\mathcal{A}_T$ defined by Eqs.~(\ref{eq:qcapeps}),(\ref{eq:ringblock}). Note that $\mathcal{A}_T$ is a different MPS tensor for each $N$. In particular, whether or not $\mathcal{A}_T$ is injective can depend on $N$. In this section only, we make the $N$ dependence explicit by using the notation $\mathcal{A}_{T,N}$ to indicate that $\mathcal{A}_T$ is defined on a ring of circumference $N$. We now prove the following result.
\begin{Prop} \label{prop:injective}
$\mathcal{A}_{T,N}$ is injective for all $N$ if and only if $t_{12}=1$, where $t$ is the polynomial matrix representation of the CQCA $T$.
\end{Prop}
The proof follows directly from the following two lemmas:
\begin{Lemma} \label{lemma:1}
$\mathcal{A}_{T,N}$ is injective if and only if the set of operators $Z_i$ and $T(Z_i)$, $i=1,\dots, N$, and their products forms a basis for all $2^N\times 2^N$ matrices.
\end{Lemma}
\begin{proof}
Recall that injectivity means that, for sufficiently large $l$, the set of products $\{\mathcal{A}_{T,N}^{\mathbf{j}_1}\dots\mathcal A_{T,N}^{\mathbf{j}_l}\}_{\mathbf{j}_1,\dots,\mathbf{j}_l}$ spans the space of all $2^N\times 2^N$ matrices. Since we have $\mathcal{A}^\mathbf{j}_{T,N}=C^\mathbf{j}_N T$ with $C^\mathbf{j}_N:=\bigotimes_{i=1}^N (Z_i)^{j_i}$, we can rewrite this set as $\{C_N^{\mathbf{j}_1} T(C_N^{\mathbf{j}_2})\dots T^{l-1} (C_N^{\mathbf{j}_l}) T^l\}_{\mathbf{j}_1,\dots,\mathbf{j}_l}$ where $T^k(C_N^\mathbf{j})=T^k C_N^\mathbf{j} T^{\dagger k}$. When checking for injectivity, the $T^l$ at the end can be ignored, since it does not affect the rank of the span. Using Eq.~(\ref{eq:cayley}), we can always reduce powers of $T$, such that $T^k(C_N^\mathbf{j})$ can always be expressed as a product of $C_N^\mathbf{j}$, $T(C_N^\mathbf{j})$ and their translations. Therefore, injectivity in the present scenario means that the operators $Z_i$, $T(Z_i)$ and their products form a basis. 
\end{proof}
In particular, Lemma \ref{lemma:1} says that $\mathcal{A}_{T,N}$ is injective if and only if it is possible to obtain the operator $X_0$ as a product of $Z_i$ and $T(Z_i)$, with $X_i$ for $i\neq 0$ following from translation invariance. 

To facilitate the rest of the proof, let us move to the Laurent polynomial notation. Recall that, when considering a ring of circumference $N$, we identify $u^N=1$. We also work over the field $\mathbb{Z}_2$, so all coefficients of the polynomials will be taken modulo 2. In what follows, we use congruency ($\equiv$) to indicate equality subject to these identifications.

The condition that we can express $X_0$ as a product of $Z_i$ and $T(Z_i)$ translates into the existence of Laurent polynomials $\bm{\xi}^Z$ and $\bm{\zeta}^Z$ such that 
\begin{equation}
\begin{pmatrix} 0 \\ \bm{\xi}^Z \end{pmatrix} + t\begin{pmatrix} 0 \\ \bm{\zeta}^Z \end{pmatrix}\equiv \begin{pmatrix} 1 \\ 0 \end{pmatrix}.
\end{equation}
In particular, we have,
\begin{equation}
t_{12}\bm{\zeta}^Z\equiv 1.
\end{equation}
That is, $t_{12}$ is invertible subject to the identification $u^N=1$, which we call $N$-invertible. This can only be true for all $N$ in the trivial case, as stated by the next Lemma,
\begin{Lemma}
A symmetric Laurent polynomial $p(u)$ with coefficients in $\mathbb{Z}_2$ is $N$-invertible for all $N$ if and only if $p(u)=1$.
\end{Lemma}
\begin{proof}
Clearly, $p(u)=1$ is $N$-invertible for all $N$. For the converse we use Ref.~\cite{Voorhees1994}, which shows that a polynomial $p(u)$ over $\mathbb{Z}_2$ is $N$-invertible if and only if no $N$-th root of unity $\omega$ is a root of $p(u)$, \textit{i.e.} $p(\omega)\equiv 0$. To finish the proof, We will first show that $p(u)$ has such a root on the unit circle. Then, we will show that it must be an $N$-th root of unity for some $N$. 

By assumption, $p(u)$ is a symmetric Laurent polynomial over $\mathbb{Z}_2$, meaning it can be written in the following form,
\begin{equation}
p(u)=\beta +\sum_k \alpha_k(u^k+u^{-k})
\end{equation}
for some $\beta,\alpha_k\in\{0,1\}$. If $\beta=0$, then $p(1)\equiv 0$, and therefore 1 is trivially a root, so assume $\beta=1$. If we write $u=e^{i\phi}$ for some angle $\phi$, then we have,
\begin{equation}
p(e^{i\phi})=1+2\sum_{k} \alpha_k \cos(k\phi).
\end{equation}
Let $K$ denote the number of non-zero $\alpha_k$ (which is finite), giving $p(1)=2K+1$. Let $\theta=\frac{\pi}{k}$ for some $k$ with $\alpha_k\neq 0$, giving $p(e^{i\theta})\leq 2K-1$. Then, since $p(e^{i\phi})$ is continuous, there must exist some $\psi\in(0,\theta)$ such that $p(e^{i\psi})=2K\equiv 0$. Calling $\omega=e^{i\psi}$, we have a root on the unit circle.

Next, notice that $p(u^{2^n})\equiv p(u)^{2^n}$ for any $n\in \mathbb{N}$. This holds because our coefficients are in $\mathbb{Z}_2$, so the cross terms cancel out each time $p(u)$ is squared. Hence we have $p(\omega^{2^n})=p(\omega)^{2^n}\equiv 0$, so $ \omega^{2^n}$ is also a root for all $n\in \mathbb{N}$. But there are only a finite number of roots on the unit circle (since $p(u)$ is bounded on it), so we must have $\omega^{2^n}=\omega^{2^m}$ for some $m<n$. Thus $\omega$ is an $N$-th root of unity for $N=k(2^n-2^m)$ for any $k\in\mathbb{N}$, and therefore $p(u)$ is not $N$-invertible.
\end{proof}

Applying this Lemma to $t_{12}$, which is by definition a symmetric Laurent polynomial, we see that if we want $\mathcal{A}_{T,N}$ to be injective for all $N$, we need to fix $t_{12}=1$. By examining the stabilizer representation of our fixed-point PEPS derived in Appendix \ref{sec:stabilizers}, it is clear that the PEPS depends only on $t_{12}$ and $\mathrm{Tr}(t)$. Hence, if we set $t_{12}=1$, the only meaningful degree of freedom left is $\mathrm{Tr}(t)$, which we can obtain by setting $t_{11}=\mathrm{Tr}(t)$ and $t_{22}=0$. Finally, our condition that $t$ has unit determinant enforces $t_{21}=1$, and we are left with Eq.~(\ref{eq:simpleqca2}).

\subsection{Properties of fixed-point PEPS with simple CQCA}

Here we note some properties of fixed-point PEPS defined by simple CQCA. 

\vspace{5mm}
\noindent
\textit{(i)} The PEPS correspond to graph states defined by some underlying graph \cite{Hein2004}, up to global application of the unitary $S$ (Eq.~(\ref{eq:sgate})). This can be clearly seen by examining the stabilizer representation of the state derived in Appendix \ref{sec:stabilizers}. 

\vspace{5mm}
\noindent
\textit{(ii)} The $L$-cycle symmetry representation $\bm{\xi}\mapsto U_T(\bm{\xi})$ defined by Eqs.~(\ref{eq:onsitesymm}) and (\ref{eq:cyclesymmdef}) is always a faithful representation of $\mathbb{Z}_2^N\times \mathbb{Z}_2^N$. To see this, suppose that $U_T(\bm{\xi})=\mathbb{I}$ for some $\bm{\xi}\in \mathbb{Z}_2^N\times\mathbb{Z}_2^N$. In particular, this implies that $u(\bm{\xi})=u(t\bm{\xi})=\mathbb{I}$. Since $u(\bm{\xi})=\bigotimes_{i=1}^N X_i^{\xi^X_i}$, we see that $\xi^X_i=0$,  $\forall i$. Using this and Eq.~(\ref{eq:simpleqca2}), we have that $u(t\bm{\xi})=\bigotimes_{i=1}^N X_i^{\xi^Z_i}$, which then implies $\xi^Z_i=0$,  $\forall i$. So $\bm{\xi}=0$ and thus $\bm{\xi} \mapsto U_T(\bm{\xi})$ is a faithful representation.

\vspace{5mm}
\noindent
\textit{(iii)} In Sec. \ref{sec:computation} we introduce an altered construction of fixed-point PEPS from CQCA that uses a two-qubit unit cell (Eq.~(\ref{eq:2qubitqcapeps})). This construction subsumes the original one, in that all fixed-point PEPS defined by Eq.~(\ref{eq:qcapeps}) with a simple CQCA can also be constructed by Eq.~(\ref{eq:2qubitqcapeps}). To see this, note that the fixed-point PEPS defined by Eq.~(\ref{eq:qcapeps}) can be written as an MPS as follows,
\begin{equation}
{|\psi_T\rangle}=\sum_{\mathbf{j}_1,\dots\mathbf{j}_M} \mathrm{Tr}\left( Z^{\mathbf{j}_1}T \dots Z^{\mathbf{j}_M}T \right) {|\mathbf{j}_1,\dots,\mathbf{j}_M\rangle}.
\end{equation}
Where $P^{\mathbf{j}}=\bigotimes_{i=1}^N (P_i)^{j_i}$ for a Pauli operator $P=X,Y,Z$. If we commute half of the $Z^{\mathbf{j}}$ past the neighbouring CQCA $T$ on the right, and we use the fact that $T(Z^{\mathbf{j}})=X^{\mathbf{j}}$ for simple CQCA, we get
\begin{align}
&{|\psi_T\rangle}= \nonumber \\
&\sum_{\mathbf{j}_1\dots\mathbf{j}_M} \mathrm{Tr}\left( X^{\mathbf{j}_{M}} Z^{\mathbf{j}_1}T^2  \dots X^{\mathbf{j}_{M-2}} Z^{\mathbf{j}_{M-1}}T^2 \right)
{|\mathbf{j}_1,\dots,\mathbf{j}_M\rangle}.
\end{align}
If we group the $i$-th qubits of columns $2k$ and $2k+1$ together, this state is exactly the fixed-point PEPS defined by Eq.~(\ref{eq:2qubitqcapeps}) with CQCA $T^2$.

\section{Proof of Theorem \ref{theorem:1}} \label{app:cycle}

Here we show how to constrain the MPS tensors of states which have 1D SPT order with respect to $L$-cycle symmetries. We begin with a state whose MPS tensor is injective and satisfies Eq.~(\ref{eq:blockcyclesymm}), where $V(g)$ is a projective representation of a finite Abelian group $G$ with maximally non-commutative cocycle $\omega$. Using this, we would like to prove Eq.~(\ref{eq:qcamps}).

First, we prove Eq.~(\ref{eq:singletenssycle}). In general, if a state is invariant under the $L$-cycle symmetry in Eq.~(\ref{eq:cyclesymm}), then we have~\cite{Molnar2018},
\begin{equation} \label{eq:fundtheorem}
\includegraphics[scale=1,valign=c,raise=0.28cm]{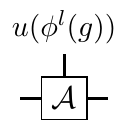}=\includegraphics[scale=1,valign=c]{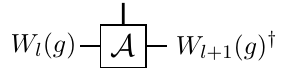},
\end{equation}
for some unitaries $W_l(g)$ ($l=0,\dots,L$) with $W_0(g)=W_L(g)$. This equation, and all others that follow in this section, hold for all $g\in G$. From Eq.~(\ref{eq:fundtheorem}), we get,
\begin{align} 
&\hspace{1.01cm} \includegraphics[scale=0.95,valign=c,raise=0.33cm]{blockcyclesymm_a.pdf} \nonumber \\
&\centerline{=} \nonumber \\
&\includegraphics[scale=0.95,valign=c]{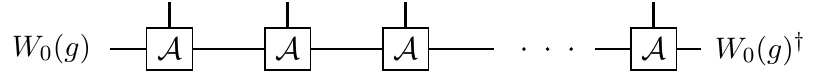}.
\end{align}
If the tensor is injective, the only way this equation and Eq.~(\ref{eq:blockcyclesymm}) can both be true is if $W_0(g) \propto V(g)$. 
If we connect the relations in Eq.~(\ref{eq:fundtheorem}) in a different order, we find,
\begin{align} 
&\hspace{1.01cm} \includegraphics[scale=0.95,valign=c,raise=0.33cm]{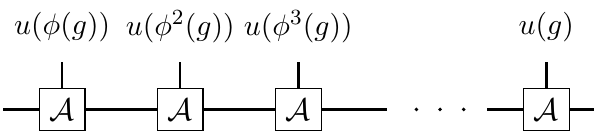} \nonumber \\
&\centerline{=}\nonumber \\
&\includegraphics[scale=0.95,valign=c]{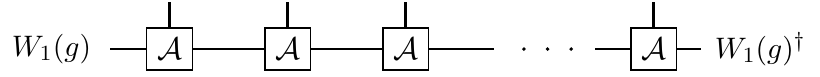},
\end{align}
where we have used the fact that $\phi^L(g)=g$. Now compare this to Eq.~(\ref{eq:blockcyclesymm}) with $g$ replaced by $\phi(g)$. Again, by injectivity, we find that $W_1(g)\propto V(\phi(g))$. So, for some scalar $\lambda (g)$, we have,
\begin{equation} \label{eq:singletenscycle2}
\includegraphics[scale=1,valign=c,raise=0.28cm]{singletenscycle_a.pdf}=\includegraphics[scale=1,valign=c]{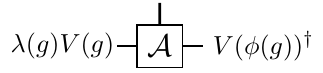}.
\end{equation}

Now, we would like to use Eq.~(\ref{eq:singletenscycle2}) to prove Eq.~(\ref{eq:qcamps}). This proof is similar to that of Theorem 1 in Ref.~\cite{Else2012}, adapted to $L$-cycle symmetries. First, we invoke maximal non-commutativity of the cocycle $\omega$, which says that there is only one irreducible representation (irrep) with cocycle $\omega$, up to unitary equivalence (\cite{Berkovich1998}, Theorem 6.39). Calling this irrep $\widetilde{V}(g)$, this implies that $V(g)$, being a reducible representation, must have the form $\mathbb{I}\otimes \widetilde{V}(g)$ in some basis. 

Then, if $g\mapsto u(g)$ is a linear representation, and $\widetilde{V}(g)$ is an irrep with cocycle $\omega$, Eq.~(\ref{eq:singletenscycle2}) tells us that the cocycle of $\widetilde{V}(\phi(g))$ must be in the same cohomology class as $\omega$. That is, there exist phases $\gamma(g)$ such that $\gamma(g) \widetilde{V}(\phi(g))$ has cocycle $\omega$. Hence, by maximal non-commutativity, we have $\gamma(g) \widetilde{V}(\phi(g))=\Phi^\dagger \widetilde{V}(g) \Phi$ for some unitary $\Phi$. Then it follows from Eq.~(\ref{eq:singletenscycle2}) that $g\mapsto \lambda(g)\gamma(g)$ is a 1D representation of $G$. Then, from Ref.~\cite{Else2012}, we can find a unitary  $\Lambda$ such that $\overline{\lambda(g)\gamma(g)} \widetilde{V}(g)=\Lambda^\dagger \widetilde{V}(g) \Lambda$.

With the above we can now rewrite Eq.~(\ref{eq:singletenscycle2}) as,
\begin{equation}
\includegraphics[scale=1,valign=c,raise=0.26cm]{singletenscycle_a.pdf}\ =\includegraphics[scale=1,valign=c]{singletenscycle3_b.pdf}.
\end{equation}
If we multiply each side by $\mathbb{I}\otimes \Phi^\dagger\Lambda^\dagger$ on the right, then we get a relation for the composite tensor $\mathcal{A}^j(\mathbb{I}\otimes \Phi^\dagger\Lambda^\dagger)$. With this relation, we can continue as in Theorem 1 of Ref.~\cite{Else2012} to obtain,
\begin{equation}
\mathcal{A}^j(\mathbb{I}\otimes \Phi^\dagger\Lambda^\dagger)=B^j\otimes C^j ,
\end{equation}
hence,
\begin{equation} \label{eq:qcamps2}
\mathcal{A}^j=B^j\otimes (C^j\Lambda \Phi) .
\end{equation}
This expression holds with respect to the basis of physical spins $\{|j\rangle\}$ that diagonalizes $u(g)$, such that $u(g){|j\rangle}=\chi_j(g){|j\rangle}$ where $g\mapsto \chi_j(g)$ are 1D representations of $G$. We then have $C^j=\widetilde{V}(g_j)$ where $g_j$ are defined uniquely by the relation $\chi_j(g) \widetilde{V}(g)=\widetilde{V}(g_j)^\dagger \widetilde{V}(g)\widetilde{V}(g_j)$. $B^j$ are unconstrained tensors and can vary throughout the SPT phase. Since $ \lambda$ is not uniformly defined throughout the SPT phase, $\Lambda$ can also vary throughout the phase. 

We would like to rewrite Eq.~(\ref{eq:qcamps2}) in such a way that we separate the parts which are universal throughout the phase, namely $\Phi$ and $C^j$, from the rest. We can accomplish this in the following way. On each block of $L$ sites, we push the matrix $\Lambda$ through the tensors to the end of the block,
\begin{align} \label{eq:blockti}
&B^{j_1}\dots B^{j_L}\otimes C^{j_1}\Lambda \Phi \dots C^{j_L} \Lambda\Phi= \nonumber \\
&B^{j_1}_{[1]} \dots B^{j_L}_{[L]}\otimes C^{j_1} \Phi \dots C^{j_L}\Phi  \Lambda_*.
\end{align}
$\Lambda$ commutes with all $\widetilde{V}(g_j)$ up to a phase. So when $\Lambda$ passes through the tensor $\mathcal{A}^j$, it leaves behind a phase, which is absorbed into the tensor $B^j$. The tensors $B^j_{[l]}$ are the original tensors $B^j$ along with these phases, which are in general not translationally invariant within a block (We do, however, have $B^j_{[l]}=B^j_{[l+L]}$, so the representation is invariant under translation by $L$ sites). The matrix $\Lambda_*:=  \Lambda\Phi^{L-1\dagger}\Lambda\Phi^{L-1}\dots \Phi^\dagger\Lambda\Phi$ is a scalar matrix and can thus be removed from the above expression. To see this, note that Eq.~(\ref{eq:blockcyclesymm}) says that $\lambda (g)$ must satisfy the constraint 
\begin{equation}
\lambda_*(g):=  \lambda(g)\lambda(\phi(g))\dots\lambda(\phi^{L-1}(g))=1. 
\end{equation}By definition, we have $\overline{\lambda_*(g)} \widetilde{V}(g)=\Lambda_*^{\dagger} \widetilde{V}(g) \Lambda_*$. Since $\lambda^*(g)=1$, we have $\Lambda_*\propto \mathbb{I}$ by Schur's lemma. The proportionality constant can be absorbed into one of the tensors $B^j_{[l]}$ and therefore $\Lambda_*$ can be removed from Eq.~(\ref{eq:blockti}). Hence, we can represent our state by an MPS of the form
\begin{equation} \label{eq:qcamps3}
A_{[l]}^j=B_{[l]}^j\otimes (C^j\Phi),
\end{equation}
as desired. $\Box$

\section{Period of fractal CQCA} \label{app:fractal}

\begin{table}
\centering
\vspace{5mm}
\begin{tabular}{|c|c|c|c|c|c|}
\hline
\textbf{N}  & L  & \textbf{N}  & L    & \textbf{N}  & L     \\ \hline
\textbf{2}  & 3  & \textbf{18} & 84   & \textbf{34} & 510   \\
\textbf{4}  & 6  & \textbf{20} & 60   & \textbf{36} & 168   \\
\textbf{6}  & 12 & \textbf{22} & 186  & \textbf{38} & 1026  \\
\textbf{8}  & 12 & \textbf{24} & 48   & \textbf{40} & 120   \\
\textbf{10} & 30 & \textbf{26} & 126  & \textbf{42} & 2340  \\
\textbf{12} & 24 & \textbf{28} & 36   & \textbf{44} & 372   \\
\textbf{14} & 18 & \textbf{30} & 1020 & \textbf{46} & 12282 \\
\textbf{16} & 24 & \textbf{32} & 48   & \textbf{48} & 96    \\ \hline
\end{tabular}
\caption{The period $L$ as a function of the circumference $N$ for the fractal CQCA $T_f$ in Eq.~(\ref{eq:tf}). For $N=2^k$, the relationship is linear. In general, $L$ appears to grow with an exponential envelope.}
\label{tab:period}
\end{table}

Here, we look at the period of fractal CQCA on rings of circumference $N$. This is in general a difficult problem to solve, and the period $L$ can appear to grow exponentially with $N$ (see Table \ref{tab:period} for an example). Nevertheless, we prove that, for certain values of $N$, the period can be shown to have a linear relation in general. Specifically, we prove that, for a fractal CQCA $T$ on a ring of circumference $N=2^k$, the period is either $L=N=2^k$ or $L=\frac{3}{2}N=2^k+2^{k-1}$. Throughout this proof, we will be taking powers of polynomials like $\sum_k c_k u^k$. As in Appendix \ref{app:simpleqca}, we can write $(\sum_k c_k u^k)^{2^n}=\sum_k c_k (u^k)^{2^n}$ $\forall n\in \mathbb{N}$. This simplifies the following calculations significantly.

Let 
\begin{equation}
\gamma:=\mathrm{Tr}(t)=\beta +\sum_i \alpha_i(u^i+u^{-i}),
\end{equation}
for $\alpha_i,\beta\in\{0,1\}$. With the periodic boundary conditions, we identify $u^{-2^{k-1}}$ and $u^{2^{k-1}}$. Then we have 
\begin{equation}
\gamma^{2^{k-1}}=\beta +\sum_i \alpha^i((u^{2^{k-1}})^i+(u^{-2^{k-1}})^i)=\beta.
\end{equation}
The proof now splits into two cases:

\textit{Case 1: $\beta=0$}. From Eq.~(\ref{eq:cayley}), we have 
\begin{equation}
t^{2^k}=(t^{2})^{2^{k-1}}=(\gamma t+\mathbb{I})^{2^{k-1}}=\gamma^{2^{k-1}}t^{2^{k-1}}+\mathbb{I},
\end{equation}
where cross terms again cancel out. Since $\gamma^{2^{k-1}}=\beta=0$, we have $t^{2^k}=\mathbb{I}$, showing $L=2^k=N$. $\Box$

\textit{Case 2: $\beta=1$}. We need a formula for $t^{2^k+2^{k-1}}$. The result, which we will prove by induction, is
\begin{align} \label{eq:fracperiod}
&t^{2^k+2^{k-1}}= \nonumber\\
&(1+\gamma^{2^k})\gamma^{2^{k-1}-1} t+\left[(1+\gamma^{2^k})p_k(\gamma)+\gamma^{2^{k-1}}\right]\mathbb{I}.\nonumber
\end{align}
Therein, $\gamma=\mathrm{Tr}(t)$, and $p_k(\gamma)$ is some polynomial in $\gamma$. The expression is simple to confirm for $k=1$ using Eq.~(\ref{eq:cayley}), where $p_1(\gamma)=0$. Now assume it is true for $k$. Then we have
\begin{align}
&t^{2^{k+1}+2^{k}}=(t^{2^k+2^{k-1}})^2= \nonumber\\
&(1+\gamma^{2^{k+1}})\gamma^{2^k-2}t^2+\left[(1+\gamma^{2k})p_k(\gamma)+\gamma^{2^{k-1}}\right]^2\mathbb{I}.\nonumber
\end{align}
Now, we apply Eq.~(\ref{eq:cayley}) to find
\begin{align}
&t^{2^{k+1}+2^{k}}=\nonumber \\
&(1+\gamma^{2^{k+1}})\gamma^{2^{k}-1} t+\left[(1+\gamma^{2^{k+1}})p_{k+1}(\gamma)+\gamma^{2^{k}}\right]\mathbb{I},\nonumber
\end{align}
which is the desired expression where $p_{k+1}(\gamma)=p^2_k(\gamma)+\gamma^{2^{k}-2}$. So the formula holds for all $k$. Now, since $\gamma^{2^{k-1}}=\beta=1$, we get $t^{2^k+2^{k-1}}=\mathbb{I}$, so $L={2^k+2^{k-1}}=\frac{3}{2}N$. $\Box$

\printbibliography


\end{document}